\documentclass{scrartcl}

\usepackage{algorithm}
\usepackage[noEnd,commentColor=black,endLComment=]{algpseudocodex}
\usepackage{amsmath}
\usepackage{amssymb}
\usepackage{amsthm}
\usepackage{booktabs}
\usepackage[shortlabels]{enumitem}
\usepackage{textcomp}
\usepackage{thmtools}
\usepackage{tikz}
\usepackage{pgfplots}
\usepackage{subcaption}
\usepackage{xcolor}

\usepackage{url}
\usepackage[hypertexnames=false]{hyperref}

\usepackage[capitalize]{cleveref}
\crefname{section}{section}{sections} 
\crefname{page}{page}{pages} 
\crefname{table}{table}{tables} 
\crefname{figure}{fig.}{figs.} 

\hypersetup{colorlinks}

\usetikzlibrary{calc,shapes}
\pgfplotsset{compat=1.16}

\newcommand{\N}{\mathbb{N}}

\newcommand{\fO}{\mathcal{O}}

\newcommand{\bd}{\partial}
\newcommand{\bdsize}[1]{|\bd(#1)|}
\newcommand{\cut}{c}

\newcommand{\tparent}{\mathtt{parent}}
\newcommand{\tdsepchild}{\mathtt{dsep\_child}}
\newcommand{\tisepchild}{\mathtt{isep\_child}}
\newcommand{\tpdist}{\mathtt{pdist}}
\newcommand{\tadist}{\mathtt{adist}}
\newcommand{\tdroot}{\mathtt{droot}}
\newcommand{\tsplayStep}{\mathtt{splay\_step}}

\newcommand{\tNodeToRoot}[1]{\Call{NodeToRoot}{#1}}
\newcommand{\tNodeBelowRoot}[1]{\Call{NodeBelowRoot}{#1}}

\newcommand{\opLink}[1]{\Call{Link}{#1}}
\newcommand{\opCut}[1]{\Call{Cut}{#1}}
\newcommand{\opPathWeight}[1]{\Call{PathWeight}{#1}}
\newcommand{\opFindRoot}[1]{\Call{FindRoot}{#1}}
\newcommand{\opEvert}[1]{\Call{Evert}{#1}}
\newcommand{\opLCA}[1]{\Call{LCA}{#1}}

\newcommand{\implfile}[1]{\path{#1}}
\newcommand{\implitem}[1]{\path{#1}}
\newcommand{\dsimpl}[1]{\texttt{#1}}

\newtheorem{theorem}{Theorem}[section]
\newtheorem{lemma}[theorem]{Lemma}
\newtheorem{observation}[theorem]{Observation}

\graphicspath{{./figures/}}

\begin{document}

\title{Fast and simple unrooted dynamic forests}
\author{Benjamin Aram Berendsohn\thanks{Freie Universit\"at Berlin, Germany. Email: \texttt{benjamin.berendsohn@fu-berlin.de}. Supported by DFG Grant \mbox{KO~6140/1-2}.}}

\date{\vspace{-5ex}} 

\maketitle

\begin{abstract}
	A \emph{dynamic forest} data structure maintains a forest (and associated data like edge weights) under edge insertions and deletions. Dynamic forests are widely used to solve online and offline graph problems. Well-known examples of dynamic forest data structures are \emph{link-cut trees} \cite{SleatorTarjan1983} and \emph{top trees} \cite{AlstrupEtAl2005}, both of which need $\fO(\log n)$ time per operation. While top trees are more flexible and arguably easier to use, link-cut trees are faster in practice~\cite{TarjanWerneck2010}.
	
	In this paper, we propose an alternative to link-cut trees. Our data structure is based on \emph{search trees on trees} (STTs, also known as \emph{elimination trees}) and an STT algorithm~\cite{BerendsohnKozma2022a} based on the classical Splay trees~\cite{SleatorTarjan1985}.
	While link-cut trees maintain a hierarchy of binary search trees, we maintain a single STT. Most of the complexity of our data structure lies in the implementation of the \emph{STT rotation} primitive, which can easily be reused, simplifying the development of new STT-based approaches.
	
	We implement several variants of our data structure in the \emph{Rust} programming language, along with an implementation of link-cut trees for comparison.
	Experimental evaluation suggests that our algorithms are faster when the dynamic forest is unrooted, while link-cut trees are faster for rooted dynamic forests.
\end{abstract}

\section{Introduction}

Maintaining a dynamically changing forest along with data associated to vertices, edges, or trees is a well-known problem with a forty-year history. Sleator and Tarjan \cite{SleatorTarjan1983} first introduced a data structure (commonly called \emph{link-cut trees}) for this task, with worst-case running time $\fO( \log n )$ per operation, where $n$ is the number of vertices in the forest. The same authors proposed a simplified amortized variant of the data structure using their \emph{Splay trees} \cite{SleatorTarjan1985}. Several alternative data structures have since been proposed, including \emph{topology trees} \cite{Frederickson1985}, \emph{ET-trees} \cite{HenzingerKing1999,Tarjan1997}, \emph{RC-trees} \cite{AcarEtAl2004} and \emph{top trees} \cite{AlstrupEtAl2005,TarjanWerneck2005,HolmEtAl2023}. Top trees are more flexible and thus more widely applicable than link-cut trees. However, an experimental evaluation by Tarjan and Werneck~\cite{TarjanWerneck2010} suggests that link-cut trees, while less flexible, are faster then top trees by a factor of up to four, likely due to their relative simplicity.

Dynamic forest data structures have a large number of applications. Link-cut trees in particular have been used as a key ingredient in algorithms and data structures including, but not limited to:
minimum cut \cite{Karger2000}, maximum flow \cite{SleatorTarjan1983,GoldbergTarjan1988}, minimum-cost flow \cite{KaplanNussbaum2013}, online minimum spanning forests~\cite{Frederickson1985,EppsteinEtAl1992,EppsteinEtAl1997,CattaneoFaruoloEtAl2010}, online lowest common ancestors~\cite{SleatorTarjan1983,HarelTarjan1984}, online planarity testing~\cite{DiBattistaTamassia1996}, and geometric stabbing queries~\cite{AgarwalEtAl2012, KaplanEtAl2003}.

In this paper, we design and implement data structures that can serve as a drop-in replacement for link-cut trees, based on \emph{search trees on trees} (STTs, described in \cref{sec:intro:stt}). We experimentally compare our data structures with link-cut trees.

We focus on maintaining \emph{unrooted} forests. More precisely, our data structures maintain an edge-weighted dynamic forest and support the following three operations:
\begin{itemize}
	\itemsep0pt
	\item $\opLink{u, v, w}$ -- Adds an edge between the vertices $u$ and $v$ with weight $w$. Assumes this edge did not exist beforehand.
	\item $\opCut{u, v}$ -- Removes the edge between $u$ and $v$. Assumes this edge existed beforehand.
	\item $\opPathWeight{u, v}$ -- Returns the sum of weights of edges on the path between $u$ and $v$, or $\bot$ if $u$ and $v$ are not in the same tree.
\end{itemize}

As weights, our implementation allows arbitrary commutative monoids. For example, when using edge weights from $(\N, \max)$, the $\opPathWeight{u,v}$ method returns the maximum edge weight on the path between~$u$ and~$v$. We note that additional operations like increasing the weight of each edge on a path, or maintaining vertex weights and certain related properties are also possible, but omitted for simplicity.

\paragraph{Rooted vs.\ unrooted forests.}\label{sec:intro:rooted-unrooted}
Some applications, like the maximum flow algorithms mentioned above, require maintaining \emph{rooted} forests. In that case, a $\opFindRoot{v}$ operation is available, which returns the root of the tree containing $v$. Moreover, adding arbitrary edges is not allowed; the $\opLink{u,v}$ operation requires that $u$ is the root of its tree, and makes $v$ the parent of $u$.

The basic variant of link-cut trees maintains rooted forests. To maintain unrooted forests, an additional operation $\opEvert{v}$ is used, which makes $v$ the root of its tree, and thus enables arbitrary $\opLink{}$s. While asymptotic performance is not affected, $\opEvert{}$ does come with additional bookkeeping that may impact performance in practice.

For our data structures, the opposite is true: they ``natively'' implement unrooted forests, and maintaining rooted forests is possible only with some overhead.

Indeed, in our experimental evaluation, link-cut trees are faster for rooted forests (without evert), and our approach is faster for unrooted forests. When explicitly maintaining rooted forests with changing roots (via $\opEvert{}$), link-cut trees again appear to be slower.
We further discuss similarities and differences between link-cut trees and our data structures in \cref{sec:splaytt-vs-link-cut}.

\paragraph{Search trees on trees.}\label{sec:intro:stt}
Our approach is based on recent results~\cite{BerendsohnKozma2022a} for \emph{search trees on trees} (STTs), also known as \emph{elimination trees}. STTs are a generalization of binary search trees (BSTs) where, intuitively, the search space is a tree, and each query is a search for a vertex. We formally define STTs in \cref{sec:prelim}.

Previous papers on STTs~\cite{BoseCardinalEtAl2020,BerendsohnKozma2022a,BerendsohnEtAl2023a} have concentrated on theoretical models, where the task is to answer vertex queries, analogous to searching for keys in a BST. The underlying tree is not modified. As in the BST setting, we distinguish between the \emph{static} and \emph{dynamic} STT model, both of which we will now briefly describe.

In the static model, we are given an input distribution and need to build an STT that can answer queries with low expected cost (according to the distribution). Optimum static \emph{BSTs} on $n$ nodes can be computed in $\fO(n^2)$ time \cite{Knuth1971}, and approximated in linear time \cite{Mehlhorn1975,Mehlhorn1977}. For optimal static STTs, no exact polynomial algorithm is known, though a PTAS~\cite{BerendsohnKozma2022a} and an $\fO(n \log n)$-time 2-approximation~\cite{BerendsohnEtAl2023a} are known.

In the dynamic model, we are allowed to modify the STT after each query. These modifications are done using \emph{STT rotations}, a straight-forward generalization of BST rotations.
Typically we assume the \emph{online} dynamic model, where we are not provided with any information about the input sequence beforehand, as opposed to the \emph{offline} case, where we know the input sequence in advance.
The central open question for dynamic BSTs and STTs (called \emph{dynamic optimality}) is whether there exists a constant-competitive online algorithm, i.e., an online algorithm whose performance matches the best offline algorithm, up to a constant factor.

For BSTs, several online algorithms are conjectured to by constant-competitive, most prominently \emph{Splay}~\cite{SleatorTarjan1985} and \emph{Greedy}~\cite{Lucas1988,Munro2000}. Both of these algorithms have many useful properties, among them \emph{static optimality}, which means that they match the optimum static tree on every input sequence\footnote{Note that the static optimum ``knows'' the queries in advance.}. The best known upper bound for competitiveness is achieved by \emph{Tango trees} \cite{DemaineEtAl2007}, with a competitive ratio of $\fO(\log \log n)$.

Bose, Cardinal, Iacono, Koumoutsos, and Langerman~\cite{BoseCardinalEtAl2020} generalized Tango trees to the STT model, thus providing a $\fO(\log \log n)$-competitive algorithm, where $n$ is the number of vertices in the tree. Berendsohn and Kozma~\cite{BerendsohnKozma2022a} generalized Splay trees to STTs and proved static optimality.

\paragraph{Dynamic forests using STTs.}
A common property of dynamic BST algorithms such as Splay is that after finding a node, it is brought to the root (e.g., using the eponymous \emph{splay} operation). The same is true for Berendsohn and Kozma's Splay generalization (called \emph{SplayTT})~\cite{BerendsohnKozma2022a}. In this paper, we use SplayTT to implement dynamic forest data structures. The amortized running time per operation is $\fO(\log n)$, matching the asymptotic performance of known data structures.

We remark that many previous dynamic forest implementations are also based on or inspired by Splay, such as the amortized variant of link-cut trees \cite{SleatorTarjan1985} and some top tree implementations \cite{TarjanWerneck2005,HolmEtAl2023}. We present the first STT-based approach.

Our implementation is highly modular: it consists of (i) a basic implementation of STTs\footnote{More precisely, \emph{2-cut STTs}, defined in \cref{sec:prelim}.} and STT rotations, (ii) a routine \tNodeToRoot{} that brings an STT node to the root via rotations, and (iii) an implementation of the operations \opLink{}, \opCut{}, and \opPathWeight{} based on \tNodeToRoot{}. Most importantly, the \tNodeToRoot{} implementation can easily be replaced with a different one. We present four variants of $\tNodeToRoot{}$, three based on SplayTT, and one simpler algorithm. Future \tNodeToRoot{} implementations, developed in the dynamic STT model, would automatically provide a new dynamic forest algorithm.

We implement our data structure in the \emph{Rust} programming language.\footnote{The source code can be found at\raggedright{} \url{https://github.com/berendsohn/stt-rs}.} The modularity described above is achieved using generics, resulting in an easily extendable library. For comparison, we also implement the amortized variant of Tarjan and Sleator's link-cut trees~\cite{SleatorTarjan1983,SleatorTarjan1985}, and some very simple linear-time data structures. We experimentally compare all implementations.

In \cref{sec:prelim}, we define STTs and related concepts. In \cref{sec:2cut-impl}, we present a basic data structure that maintains a 2-cut search tree on a \emph{fixed} tree under rotations. In \cref{sec:link-cut} and \cref{sec:edge-weights}, we show how to implement the dynamic forest operations, using multiple 2-cut STTs and assuming a black-box \tNodeToRoot{} implementation. In \cref{sec:node-to-root}, we present our four \tNodeToRoot{} algorithms. In \cref{sec:experiments}, we present our experimental results, and in \cref{sec:conclusion}, we discuss our findings and propose some open questions.

\section{Preliminaries}\label{sec:prelim}

Let $G$ be a graph. The sets of vertices, resp., edges in $G$ are denoted by $V(G)$, resp., $E(G)$. The subgraph of $G$ induced by the vertex set $U \subseteq V(G)$ is denoted by~$G[U]$.

Below, we repeat definitions and observations from previous papers~\cite{BoseCardinalEtAl2020,BerendsohnKozma2022a,BerendsohnEtAl2023a}.

\paragraph{Search trees on graphs.}
A \emph{search tree} on a connected graph $G$ is a rooted tree $T$ that is constructed as follows: Choose an arbitrary vertex $r$ as the root. Then, recursively construct search trees on each connected component of $G \setminus r$ and attach them to $r$ as subtrees. We denote by $V(T)$ the set of nodes in $T$. Observe that $V(T) = V(G)$. For each node $v \in V(T)$, we denote the subtree of $T$ rooted at $v$ by $T_v$. The \emph{root path} of $v$ in $T$ is defined as the path from $v$ to the root of $T$ in $T$.

A search tree $T$ on a graph $G$ satisfies the following properties:
\begin{enumerate}[(i)]
	\item For each edge $\{u,v\} \in E(G)$, either $u$ is an ancestor of $v$ or $v$ is an ancestor of $u$;
	\item For each node $v \in V(T)$, the subgraph $G[V(T_v)]$ is connected.
\end{enumerate}

It can be shown that the above two properties in fact fully characterize search trees on~$G$.

If $G$ is a path, then search trees on $G$ essentially correspond to binary search trees on $|V(G)|$ nodes. Note, however, that children are unordered in STTs.

\paragraph{Cuts and boundaries.} 
Let $G$ be a graph and $H$ be a subgraph of $G$. The \emph{cut} of $H$, denoted $\cut_G(H)$, is the set of edges between $H$ and $G \setminus H$. The \emph{outer boundary} of $H$, denoted $\bd_G(H)$, is the set of vertices in $G \setminus H$ that are adjacent to some vertex in $H$.
Note that $|\bd_G(H)| \le |\cut_G(H)|$ and $|\bd_G(H)| = |\cut_G(H)|$ if $G$ is a tree.
If $T$ is a search tree on $G$ and $v \in V(T)$, we write $\bd(T_v) = \bd_G(V(T_v))$ for short.
The following observations will be useful later.

\begin{observation}\label{p:boundary-rules}
	Let $p$ be a node in an STT $T$, and let $v, v'$ be children of $p$. Then:
	\begin{enumerate}[(i)]
		\item $p \in \bd(T_v)$;
		\item $\bd(T_v) \subseteq \bd(T_p) \cup \{p\}$;
		\item $\bd(T_v) \cap \bd(T_{v'}) = \{p\}$.\label{item:boundary-rules:child-itersec}
	\end{enumerate}
\end{observation}

We remark that (\ref{item:boundary-rules:child-itersec}) follows from the lack of cycles in $G$ and is not true for search trees on graphs in general.

\paragraph{Rotations.}
Let $v$ be a node in a search tree $T$ on a graph $G$, and let $p$ be the parent of $v$. A \emph{rotation} of $v$ with its parent (also called a \emph{rotation at $v$}) is performed as follows. Make $p$ a child of $v$ and make $v$ a child of the previous parent of $p$, if it exists (otherwise, make $v$ the root). Then, every child of $v$ with $p \in \bd(T_v)$ is made a child of $p$. Observe that if $G$ is a tree, then only one child of $v$ can change parent like this, otherwise $p$ is part of a cycle. \Cref{fig:stt-rot} shows a rotation in an STT.

\tikzset{
	vertex/.style={fill, circle, inner sep = 1.2pt},
	component/.style={draw, circle, minimum width = 8mm, minimum height = 8mm, inner sep = 0pt},
	empty triangle/.style = {regular polygon, regular polygon sides = 3, inner sep = -1pt, minimum size = 8mm},
	triangle/.style = {empty triangle, draw},
	nice dash/.style={dash pattern=on 1pt off 1pt},
	out dash/.style={nice dash},
	markTri/.style={fill={gray!30!white}}
}

\begin{figure}
	\centering
	\begin{tikzpicture}[sibling distance = 11mm, level distance = 13mm, triangle/.append style={font=\small}]
		\begin{scope}[local bounding box=T1]
			\node[vertex] (p) {}
				child[child anchor = north]{ node[triangle] (A1) {\strut$A_1$} }
				child[child anchor = north]{ node[triangle] (A2) {\strut$A_j$} }
				child{ node[vertex] (c) {}
					child[child anchor = north]{ node[triangle, markTri] (B) {\strut$B$} }
					child[child anchor = north]{ node[triangle] (C1) {\strut$C_1$} }
					child[child anchor = north]{ node[triangle] (C2) {\strut$C_\ell$} }
				};
			\draw[out dash] (p) -- ($(p)+(0,0.3)$);
			\node[right=0.5mm] at (p) {$u$};
			\node[above right] at (c) {$v$};
			\node at ($(A1)!0.5!(A2)+(0,0.2)$) {$\dots$};
			\node at ($(C1)!0.5!(C2)+(0,0.2)$) {$\dots$};
		\end{scope}
		
		\begin{scope}[xshift=60mm,local bounding box=T2]
			\node[vertex] (c) {}
				child{ node[vertex] (p) {}
					child[child anchor = north]{ node[triangle] (A1) {\strut$A_1$} }
					child[child anchor = north]{ node[triangle] (A2) {\strut$A_j$} }
					child[child anchor = north]{ node[triangle, markTri] (B) {\strut$B$} }
				}
				child[child anchor = north]{ node[triangle] (C1) {\strut$C_1$} }
				child[child anchor = north]{ node[triangle] (C2) {\strut$C_\ell$} };
			\draw[out dash] (c) -- ($(c)+(0,0.3)$);
			\node[left=0.5mm] at (c) {$v$};
			\node[above left] at (p) {$u$};
			\node at ($(A1)!0.5!(A2)+(0,0.2)$) {$\dots$};
			\node at ($(C1)!0.5!(C2)+(0,0.2)$) {$\dots$};
		\end{scope}
		
		\draw[-latex] (T1) -- (T2);
	\end{tikzpicture}
	\hspace{10mm}
	\begin{tikzpicture}[xscale=0.7, yscale=0.75]
		\node[component] (A1) at (0,0) {$A_k\strut$};
		\node[] (ARem) at (0,1+0.1) {$\vdots$};
		\node[component] (A2) at (0,2) {$A_1\strut$};
		\node[vertex] (p) at (1,1) {};
		\node[above=1mm] at (p) {$u$};
		\node[component] (B) at (2,1) {$B\strut$};
		\node[vertex] (c) at (3,1) {};
		\node[above=1mm] at (c) {$v$};
		\node[component] (C1) at (4,0) {$C_\ell\strut$};
		\node[] (CRem) at (4,1+0.1) {$\vdots$};
		\node[component] (C2) at (4,2) {$C_1\strut$};
		
		\draw[] (A1) -- (p) -- (B) -- (c) -- (C1);
		\draw[] (A2) -- (p);
		\draw[] (c) -- (C2);
		
		\draw[out dash] (B) -- ++(-75:8mm);
		\draw[out dash] (B) -- ++(-105:8mm);
		
		\draw[out dash] (B) -- ++(75:8mm);
		\draw[out dash] (B) -- ++(105:8mm);
		
		\draw[out dash] (C1) -- ++(-30:8mm);
		\draw[out dash] (C1) -- ++(-60:8mm);
		
		\draw[out dash] (C2) -- ++(30:8mm);
		\draw[out dash] (C2) -- ++(60:8mm);
		
		\draw[out dash] (A1) -- ++(-120:8mm);
		\draw[out dash] (A1) -- ++(-150:8mm);
		
		\draw[out dash] (A2) -- ++(120:8mm);
		\draw[out dash] (A2) -- ++(150:8mm);
	\end{tikzpicture}
	\caption{An STT rotation (left) and the relevant parts of the underlying tree (right).}\label{fig:stt-rot}
\end{figure}
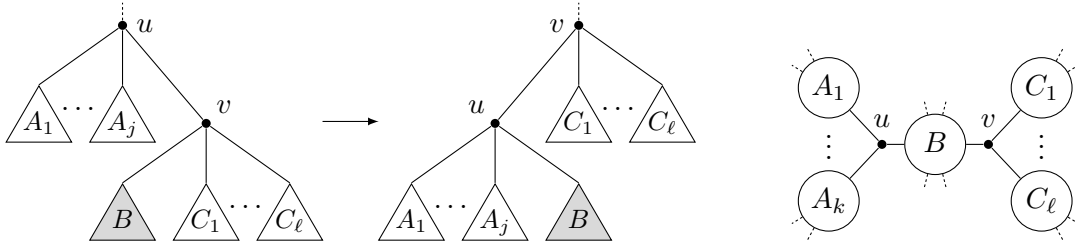

\paragraph{k-cut search trees.}
Let $T$ be a search tree on a graph. A node $v \in V(T)$ is called \emph{$k$-cut} if $|\bd(T_v)| \le k$. In an STT, this means that at most $k$ edges go from $V(T_v)$ to the ancestors of $v$. The search tree $T$ itself is called $k$-cut if every node of $T$ is $k$-cut.

Observe that 1-cut STTs are simply \emph{rootings} of the underlying tree. Our data structures are based on 2-cut STTs. Intuitively, 2-cut STTs ``locally'' behave like BSTs, which allows applying familiar BST techniques. In particular, a node, its parent, and its grandparent in a 2-cut STT must lie on the same path in the underlying tree. Recall that BSTs are search trees on paths, so rotations on these three nodes will behave like BST rotations. This observation is key in Berendsohn and Kozma's Splay generalization \cite{BerendsohnKozma2022a}.
In the following, we prove a slightly stronger property.

\begin{lemma}\label{p:Steiner-closed-three-path}
	Let $v$ be a node in a 2-cut search tree $T$ on a tree $G$. Let $p$ be the parent of $v$ and let $a \in \bd(T_p)$. Then $v, p, a$ must lie on a common path in $G$ (though not necessarily in that order). 
\end{lemma}
\begin{proof}
	Suppose not. Then, there is a node $x \notin \{v,p,a\}$ such that $v, p, a$ are in pairwise distinct components of $G \setminus x$. Clearly, $x$ cannot be an ancestor of $p$, since otherwise $p$ and $v$ would be in different subtrees. Since $v$ is a child of $p$, we know that $x$ must be a descendant of~$v$. Let $c$ be the child of $v$ with $x \in V(T_c)$. We trivially have $v \in \bd(T_c)$. Since there is a path between $x$ and $p$ that does not contain $v$, we also have $p \in \bd(T_c)$. Finally, there is a path between $x$ and $a$ that does not contain $v$ or $p$, implying $a \in \bd(T_c)$. But then $\bdsize{T_c} \ge 3$, violating the 2-cut property.
\end{proof}

\section{Implementing 2-cut STTs}\label{sec:2cut-impl}

Previous works~\cite{BoseCardinalEtAl2020,BerendsohnKozma2022a,BerendsohnEtAl2023a} have not given an explicit implementation of STTs as a data structure. In this section, we show how to efficiently maintain a 2-cut search tree on a fixed tree $G$ under rotations.\footnote{See \path{stt/src/twocut/basic.rs} in the source code.}

Let $T$ be a 2-cut search tree on a tree $G$, and let $v$ be a node in $T$. We call $v$ a \emph{separator node} if $|\bd(T_v)| = 2$. Observe that a separator node $v$ is on the path (in $G$) between the two nodes in $\bd(T_v)$, hence ``separates'' them.

We call a separator node $v$ a \emph{direct separator node} if $\bd(T_v)$ contains precisely the parent and grandparent of $v$ in $T$, and an \emph{indirect separator node} otherwise.

Our representation of 2-cut STTs is based on maintaining separator and indirect separator children of nodes. It turns out that each node can have at most one separator child and one indirect separator child.

\begin{lemma}\label{p:sep-nodes}
	Each node in a 2-cut STT has up to one child that is a direct separator node, and up to one child that is an indirect separator node.
\end{lemma}
\begin{proof}
	Suppose a node $u$ has two direct separator children $v, v'$. Then $\bd(T_v) = \bd(T_{v'}) = \{u,p\}$, where $p$ is the parent of $u$. But $\bd(T_v) \cap \bd(T_{v'}) = \{u\}$ by \cref{p:boundary-rules} (\ref{item:boundary-rules:child-itersec}), a contradiction.
	
	Now suppose $u$ has two indirect separator children $v, v'$. Then $u$ has a parent $p$, and there are distinct ancestors $a_1$, $a_2$ of $p$ such that $\bd(T_v) = \{a_1,u\}$ and $\bd(T_{v'}) = \{a_2,u\}$ by \cref{p:boundary-rules} (\ref{item:boundary-rules:child-itersec}). Thus, $\{p,a_1, a_2\} \subseteq \bd(T_u)$ contradicting that $T$ is 2-cut.
\end{proof}

\Cref{p:sep-nodes} suggests the following representation of an STT $T$. For each node $v$, we store the following pointers:
\begin{itemize}
	\item $\tparent(v)$: The parent node of $v$, or $\bot$ if $v$ is the root.
	\item $\tdsepchild(v)$: The unique child of $v$ that is a direct separator node, or $\bot$ if $v$ has no such child.
	\item $\tisepchild(v)$: The unique child of $v$ that is a indirect separator node, or $\bot$ if $v$ has no such child.
\end{itemize}
With these pointers, $T$ uniquely represents $G$ (see \cref{sec:represent-proof}).

It is important to note that not all rotations are actually possible while maintaining the 2-cut property. Berendsohn and Kozma proved the characterization given below.~\cite[Lemma 5.2]{BerendsohnKozma2022a}

\begin{lemma}\label{p:rot-allowed}
	Let $T$ be an STT and let $v \in V(G)$ with parent $p \in V(G)$. Rotating $v$ with $p$ maintains the 2-cut property if and only if $|\bd(T_v)| \neq 1$ or $|\bd(T_p)| \neq 2$.
\end{lemma}

Rotations satisfying the requirements of \cref{p:rot-allowed} can be implemented in constant time using the above data structure. This (rather technical) procedure is found in \cref{sec:rot-impl}.

\section{Linking and cutting}\label{sec:link-cut}

In this section, we show how to implement the operations \opLink{} and \opCut{} to add and remove edges.\footnote{See \path{stt/src/twocut/mod.rs}} The underlying forest $G$ is maintained as a collection of 2-cut STTs, which we call a \emph{search forest}.
For a node $v$ in a search forest $F$, we denote by $F_v$ the subtree rooted at $v$.
Since we do not allow adding and removing nodes, we maintain all nodes in a fixed-size array. The structure of each STT is represented by the node pointers described in \cref{sec:2cut-impl}.

We assume that we have a black-box algorithm \tNodeToRoot{} that, given a node, brings it to the top of its tree with some sequence of rotations. (Implementations are presented in \cref{sec:node-to-root}.) We additionally assume that \tNodeToRoot{} is \emph{stable}, which essentially means that a call to \tNodeToRoot{} does not move around the previous root too much. Formally, an algorithm for \tNodeToRoot{} is called stable if, in the resulting search tree, the depth of the previous root $r$ is bounded by some constant, and all nodes on the root path of $r$ are 1-cut. The stability property simplifies the implementation of \opCut{} and \opPathWeight{}, but is not strictly necessary (see \cref{sec:non-stable}).

\begin{figure}
	\begin{minipage}{.5\columnwidth}
		\begin{algorithmic}
			\Procedure{Link}{$u,v,w$}
			\LComment{Assume $\{u,v\} \notin E(G)$}
			\State $\tNodeToRoot{u}$
			\State $\tNodeToRoot{v}$
			\State $\tparent(u) \gets v$
			\EndProcedure
		\end{algorithmic}
	\end{minipage}%
	\begin{minipage}{.5\columnwidth}
		\begin{algorithmic}
			\Procedure{Cut}{$u,v$}
			\LComment{Assume $\{u,v\} \in E(G)$}
			\State $\tNodeToRoot{u}$
			\State $\tNodeToRoot{v}$
			\State $\tparent(u) \gets \bot$
			\EndProcedure
		\end{algorithmic}
	\end{minipage}
	\caption{Pseudocode for $\opLink{}$ and $\opCut{}$.}\label{alg:link-cut}
\end{figure}

The implementations of $\opLink{}$ and $\opCut{}$ are shown in \cref{alg:link-cut}. Note that we ignore the supplied weight $w$ in $\opLink{}$ for now.

We now argue the correctness of the two procedures. Below, $G$ and $G'$ denote the underlying forest before and after the operation.
\begin{itemize}
	\item Consider a call $\opLink{u,v,w}$.
	Let $F$ be the search forest after the two calls to $\tNodeToRoot{}$, and let $F'$ the search forest after $\opLink{}$. If we only consider parent pointers, then $F'$ is clearly a valid search forest on $G'$. It remains to show that child pointers are still valid. For this, observe that for every node $x \in V(F) \setminus \{u\}$, we have $\bd(F_x) = \bd(F'_x)$, and we have $\bd(F_u) = \emptyset$, $\bd(F'_u) = \{v\}$. Thus, no node becomes a separator child or stops being one, and no direct separator node becomes an indirect one or vice versa. Thus, all child pointers stay valid.
	
	Observe that the call $\tNodeToRoot{v}$ is not necessary for correctness. However, it is important for the complexity analysis in \cref{sec:splay-analysis}.
	
	\item Now consider a call $\opCut{u,v}$.
	Again, let $F$ be the search forest after the two calls to $\tNodeToRoot{}$, and let $F'$ be the search forest after $\opCut{}$. Stability implies that $u$ is 1-cut in $F$. Since $v$ is an ancestor of $u$ and $\{u,v\} \in E(G)$, we have $\bd(F_u) = \{v\}$, implying that $v$ is the parent of $u$. Thus, setting $\tparent(u) \gets \bot$ removes the edge $\{u,v\}$ from the underlying forest. Again, the boundaries of subtrees other than $T_u$ do not change between $F$ and $F'$, implying that no further pointer changes are necessary to make $F'$ valid.
\end{itemize}

The running time of both operations is dominated by the calls to \tNodeToRoot{}, which we later show have amortized complexity $\fO(\log n)$ for our SplayTT-based variants (see \cref{sec:splay-analysis}), where $n$ is the number of vertices in the underlying forest.

\paragraph{Non-stable implementations.}\label{sec:non-stable}
Our Rust implementation supports non-stable implementations of \tNodeToRoot{}; although in this case, a second procedure $\tNodeBelowRoot{v}$ is required, which rotates $v$ directly below the current root. The implementations of \opCut{} (\cref{alg:link-cut}) and \opPathWeight{} (see \cref{sec:path-weight}) are easily adapted.

\section{Maintaining edge weights}\label{sec:edge-weights}

In this section, we show how to maintain edge weights in 2-cut STTs under rotations and implement the \opPathWeight{} operation.\footnote{In the source code, the weight update procedures are found in \path{stt/src/twocut/node_data.rs}; implementations of \opPathWeight{} (for stable and non-stable \tNodeToRoot{}) are found in \path{stt/src/twocut/mod.rs}} For simplicity, we assume that the edge weights come from a \emph{group} here instead of a monoid. A somewhat more complicated way to handle monoids is shown in \cref{sec:monoid-weights}.

Let $F$ be a 2-cut search forest on a forest $G$ with edge weights from a commutative group $(W,+)$. The weight of a path in $F$ is defined as the sum of the weights of its edges. For two vertices $u, v \in V(G)$ in the same tree, let $d(u,v)$ denote the weight of the unique path between $u$ and~$v$, i.e., the \emph{distance} between $u$ and $v$.

For each node $v$, we store a field $\tpdist(v)$ indicating the distance between $v$ and the parent of $v$ in $F$. If $v$ is the root, then $\tpdist(v) = \infty$.

\paragraph{Rotations.} Consider a rotation of $v$ with its parent $p$. Let $c$ be the direct separator child of $v$, or $\bot$ if $v$ has no direct separator child. Let $g$ be the parent of $p$, or $\bot$ if $p$ is the root. Let $F$, $F'$ be the search forest before and after the rotation. We denote by $\tpdist(\cdot)$ and $\tpdist'(\cdot)$ the values before and after the rotation.
\begin{itemize}
	\item In $F'$, the parent of $p$ is $v$, so $\tpdist'(p) = d(p,v) = \tpdist(v)$.
	\item If $p$ is the root of $F$, then $v$ is the root of $F'$, so $\tpdist'(v) = \infty$.
	\item If $p$ is not the root, then $g$ exists, and $\tpdist'(v) = d(v,g)$. Since $F$ is 2-cut, $v, p, g$ lie on a common path (by \cref{p:Steiner-closed-three-path}).
	\begin{itemize}
		\item If $v$ is between $p$ and $g$ on the path, then $v$ is a direct separator (by definition). We have $\tpdist'(v) = d(v,g) = d(p,g) - d(p,v) = \tpdist(p) - \tpdist(v)$.
		\item If $p$ is between $v$ and $g$ on the path, then $v$ is 1-cut or an indirect separator. We have $\tpdist'(v) = d(v,g) = d(v,p) + d(p,g) = \tpdist(v) + \tpdist(p)$.
		\item $g$ cannot be between $v$ and $p$, since then $v$ and $p$ needed to be in different subtrees of~$g$.
	\end{itemize}
	\item Suppose $c$ exists. Since $c$ is a direct separator in $T$, $c$ lies on the path between $v$ and $p$, so we have $\tpdist'(c) = d(c,p) = d(v,p) - d(v,c) = \tpdist(v) - \tpdist(c)$.
\end{itemize}

Since only the parents of $v$, $p$, and (possibly) $c$ change, an update procedure for $\tpdist(\cdot)$ after a rotation follows from the observations above.

\paragraph{Linking and cutting.} Aside from rotations, at the end of $\opLink{u,v,w}$, we make $v$ the parent of $u$. We can simply set $\tpdist(u) \gets w$ here. Similarly, at the end of $\opCut{u,v}$, the node $u$ is removed from its parent, and we set $\tpdist(u) \gets \infty$.

\paragraph{Computing path weight.}\label{sec:path-weight} Finally, we implement $\opPathWeight{u,v}$ as follows. First, we call $\tNodeToRoot{u}$, and then $\tNodeToRoot{v}$. Afterwards, we follow parent pointers to check whether $v$ is the root of the search tree containing $u$. If no, we return~$\bot$. If yes, we return the sum $\sum_{x \in P \setminus \{v\}} \tpdist(x)$, where $P$ is the root path of $u$.

We now argue that this procedure is correct. Let $F$ be the search forest after the two calls to \tNodeToRoot{}. Clearly, $v$ is the root of its search tree in $F$. If $u$ is in a different search tree, the algorithm correctly returns $\bot$.

Otherwise, $u$ is a descendant of $v$. Let $u = u_1, u_2, \dots, u_k = v$ be the root path of $u$ in $F$. Stability of $\tNodeToRoot{}$ implies that $u_1, u_2, \dots, u_{k-1}$ are all 1-cut. This means that for each $i \in [k-2]$, the path from any node in $V(T_{u_i})$ to any node outside of $V(T_{u_i})$ must contain $u_{i+1}$. In particular, the path from $u$ to $u_{i+2}$ contains $u_{i+1}$. Thus, by induction, the path from $u$ to $v = u_k$ contains $u_1, u_2, \dots, u_k$, in that order, and its weight is $\sum_{i=1}^{k-1} d(u_i, u_{i+1}) = \sum_{i=1}^{k-1} \tpdist(u_i)$.

The running time of \opPathWeight{} is dominated by the calls to \tNodeToRoot{}, since stability of \tNodeToRoot{} implies that $k$ is bounded.

\section{Heuristics for NodeToRoot}\label{sec:node-to-root}

\newcommand{\tsplayTo}{\mathtt{splay\_to}}
\newcommand{\tcanRotate}{\mathtt{can\_rotate}}
\newcommand{\trotate}{\mathtt{rotate}}
\newcommand{\tcanSplayStep}{\mathtt{can\_splay\_step}}
\newcommand{\tisSep}{\mathtt{is\_separator}}

\newcommand{\algMRT}{MoveToRootTT}
\newcommand{\algGSplay}{GreedySplayTT}

In the following sections, we describe multiple implementations of the \tNodeToRoot{} procedure used by \opLink{}, \opCut{}, and \opPathWeight{}.

\subsection{MoveToRootTT}\label{sec:mtr}
One of the simplest dynamic BST algorithms is the \emph{move-to-root} heuristic ~\cite{AllenMunro1978}. After finding a node $v$, it simply rotates $v$ with its parent until $v$ becomes the root.

This algorithm does not work for STTs, since not all rotations are allowed. However, if a rotation of $v$ with its parent $p$ is not allowed, then $|\bd(T_v)| = 1$ and $|\bd(T_p)| = 2$ by \cref{p:rot-allowed}, implying that $p$ is not the root and, in particular, $p$ can be rotated with its parent. Thus, we can bring $v$ to the root by repeatedly rotating at $v$ or, if that is not possible, rotating at its parent, until $v$ is the root. We call this algorithm \algMRT{} (see \cref{alg:mtr}).\footnote{Found in \path{stt/src/twocut/mtrtt.rs} in the source code.}

Observe that if the underlying tree $G$ is a path, then all rotations are possible. Thus, in this case, \algMRT{} is equivalent to the classical move-to-root algorithm. It is known that move-to-root performs poorly in the worst case, but well on uniformly random inputs \cite{AllenMunro1978}. Our experiments suggest the same for \algMRT{}. In fact, probably due to its simplicity, it outperforms more complicated algorithms on uniformly random inputs.

\begin{figure}
	\begin{minipage}{.5\textwidth}
		\begin{algorithmic}
			\Procedure{NodeToRootMTR}{$v$}
				\While{$v$ has parent $p$}
					\If{$\tcanRotate(v)$}
						\State $\trotate(v)$
					\Else
						\State $\trotate(p$)
					\EndIf
				\EndWhile
			\EndProcedure
		\end{algorithmic}
		\captionof{figure}{The \algMRT{} implementation of \tNodeToRoot{}.}\label{alg:mtr}
	\end{minipage}%
	\begin{minipage}{.5\textwidth}
		\begin{algorithmic}
			\Procedure{$\tsplayStep$}{$v$}
				\If{$v$ has no grandparent}
					\State $\trotate(v)$
				\Else
					\If{$v$ is a separator}
						\LComment{ZIG-ZAG}
						\State $\trotate(v)$
						\State $\trotate(v)$
					\Else
						\LComment{ZIG-ZIG}
						\State $\trotate(\tparent(v))$
						\State $\trotate(v)$
					\EndIf
				\EndIf
			\EndProcedure
		\end{algorithmic}
		\captionof{figure}{The $\tsplayStep$ procedure.}\label{alg:splay_step}
	\end{minipage}
\end{figure}

\subsection{SplayTT}\label{sec:splay}
In this section, we present the Splay\-TT algorithm of Berendsohn and Kozma~\cite{BerendsohnKozma2022a} and two simple variants. It is based on the classical Splay algorithm~\cite{SleatorTarjan1985}, which we describe first.

Splay can be seen as a slightly more sophisticated version of the move-to-root algorithm. After finding a node $v$, it is brought to the root by a series of calls to the procedure $\tsplayStep(v)$. If $v$ has no grandparent, then $\tsplayStep(v)$ simply rotates $v$ with its parent $p$ (this is called a ZIG step). If the value of $v$ is between the values of $p$ and its grandparent $g$, then $\tsplayStep(v)$ rotates twice at $v$ (ZIG-ZAG step). Finally, if the value of $v$ is smaller or larger than both values of $p$ and $g$, then $\tsplayStep(v)$ rotates first at $p$ and then at $v$ (ZIG-ZIG step). Afterwards, $v$ is an ancestor of both $p$ and $g$, so $v$ is eventually brought to the root.

In 2-cut STTs, $\tsplayStep(v)$ can be applied basically as-is, since \cref{p:Steiner-closed-three-path} implies that $v$, $p$, and $g$ are on a path. If $v$ is between $p$ and $g$, then we execute a ZIG-ZAG step; otherwise, we execute a ZIG-ZIG step (see \cref{alg:splay_step}). However, simply applying $\tsplayStep$ repeatedly may destroy the 2-cut property. The following lemma characterizes situations where $\tsplayStep(v)$ is allowed.

\begin{restatable}{lemma}{pSplayStepAllowed}\label{p:splay-step-allowed}
	Let $v$ be a node in an STT $T$. If $v$ is a child of the root of $T$, then $\tsplayStep(v)$ preserves the 2-cut property.
	If $v$ has a parent $p$ and a grandparent $g$, then $\tsplayStep(v)$ preserves the 2-cut property if and only if $g$ is not a separator or both $v$ and $p$ are separators.
\end{restatable}
\begin{proof}
	If $v$ is the child of the root $r$, then $\tsplayStep(v)$ performs a single rotation, which is valid (i.e., preserves the 2-cut property), since $|\bd(T_r)| = 0$.
	
	Suppose $v$ is note the child of the root, and we apply $\tsplayStep(v)$. Let $T'$ be the search tree after the first rotation, and $T''$ be the search tree after the second rotation.
	
	If $\tsplayStep(v)$ executes a ZIG-ZIG step, then it first rotates $p$ with $g$, and then $v$ with $p$. The first rotation is invalid iff $|\bd(T_p)| = 1$ and $|\bd(T_g)| = 2$. The second rotation is invalid iff $|\bd(T'_v)| = |\bd(T_v)| = 1$ and $|\bd(T'_p)| = |\bd(T_g)| = 2$. So the ZIG-ZIG step is invalid if $g$ is a separator and at least one of $v$ and $p$ is not. This is precisely the negation of the condition stated in the lemma.
	
	If $\tsplayStep(v)$ executes a ZIG-ZAG step, then it rotates twice at $v$. The first rotation is invalid iff $|\bd(T_v)| = 1$ and $|\bd(T_p)| = 2$. This can never happen, since we only execute a ZIG-ZAG step if $v$ is a separator. The second rotation is invalid iff $|\bd(T'_v)| = |\bd(T_p)| = 1$ and $|\bd(T'_g)| = |\bd(T_g)| = 2$. Assuming that $v$ is a separator, this is again the negation of the stated condition.
\end{proof}

\newcommand{\dsep}{d}
\newcommand{\ndsep}{$\neg\mathrm{d}$}
\newcommand{\isep}{i}
\newcommand{\nisep}{$\neg\mathrm{i}$}
\newcommand{\sep}{s}
\newcommand{\nsep}{$\neg\mathrm{s}$}
\newcommand{\rotlabel}[1]{\scriptsize($#1$)}
\tikzset{
	mainnode/.style={},
	labell/.style={left=-.5mm},
	labelr/.style={right=-.5mm},
	labelal/.style={left=-.5mm},
	labelar/.style={right=-.5mm},
	labelall/.style={above left=-1mm},
	labelarr/.style={above right=-1mm}
}

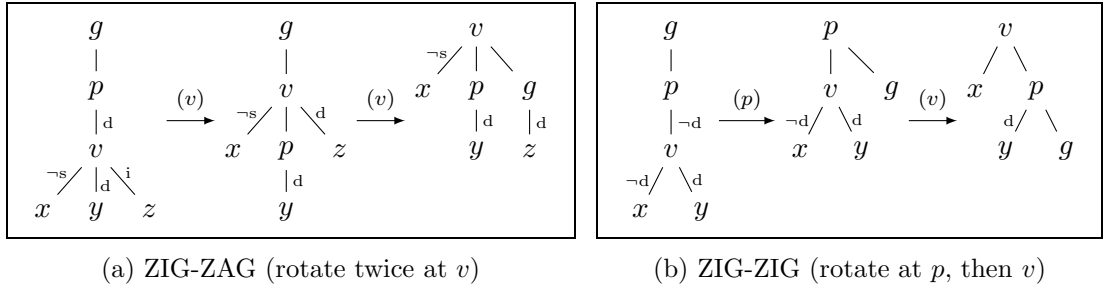
\begin{figure}
	\begin{subfigure}{.53\textwidth}
		\centering
		\fbox{
		\begin{tikzpicture}[x=3.5mm, y=8mm]
			\begin{scope}[local bounding box=zigzag1]
				\node (g) at (0,0) {$g$};
				\node (p) at (0,-1) {$p$}; \draw (g) -- (p);
				\node[mainnode] (v) at (0,-2) {$v$}; \draw (p) -- node[labelr] {\tiny \dsep} (v);
				\node (x) at (-2,-3) {$x$}; \draw (v) -- node[labelall] {\tiny \nsep} (x);
				\node (y) at (0,-3) {$y$}; \draw (v) -- node[shift={(0.35,-0.1)}] {\tiny \dsep} (y);
				\node (z) at (2,-3) {$z$}; \draw (v) -- node[labelarr] {\tiny \isep} (z);
			\end{scope}
			\begin{scope}[xshift=25mm, local bounding box=zigzag2]
				\node (g) at (0,0) {$g$};
				\node[mainnode] (v) at (0,-1) {$v$}; \draw (g) -- (v);
				\node (x) at (-2,-2) {$x$}; \draw (v) -- node[labelall] {\tiny \nsep} (x);
				\node (p) at (0,-2) {$p$}; \draw (v) -- (p);
				\node (z) at (2,-2) {$z$}; \draw (v) -- node[labelarr] {\tiny \dsep} (z);
				\node (y) at (0,-3) {$y$}; \draw (p) -- node[labelr] {\tiny \dsep} (y);
			\end{scope}
			\begin{scope}[xshift=50mm, local bounding box=zigzag3]
				\node[mainnode] (v) at (0,0) {$v$};
				\node (x) at (-2,-1) {$x$}; \draw (v) -- node[labelall] {\tiny \nsep} (x);
				\node (p) at (0,-1) {$p$}; \draw (v) -- (p);
				\node (g) at (2,-1) {$g$}; \draw (v) -- (g);
				\node (y) at (0,-2) {$y$}; \draw (p) -- node[labelr] {\tiny \dsep} (y);
				\node (z) at (2,-2) {$z$}; \draw (g) -- node[labelr] {\tiny \dsep} (z);
			\end{scope}
			\draw[-latex] ($(zigzag1.north east)!(zigzag2)!(zigzag1.south east)$) -- node[above] {\rotlabel{v}} (zigzag2);
			\draw[-latex] (zigzag2) -- node[above] {\rotlabel{v}} ($(zigzag3.north west)!(zigzag2)!(zigzag3.south west)$);
		\end{tikzpicture}
		}
		\caption{ZIG-ZAG (rotate twice at $v$)}\label{fig:zig-zag}
	\end{subfigure}%
	\begin{subfigure}{.47\textwidth}
		\centering
		\fbox{
		\begin{tikzpicture}[x=4mm, y=8mm]
			\begin{scope}[local bounding box=zigzig1]
				\node (g) at (0,0) {$g$};
				\node (p) at (0,-1) {$p$};
				\node[mainnode] (v) at (0,-2) {$v$};
				\node (x) at (-1,-3) {$x$};
				\node (y) at (1,-3) {$y$};
				\draw (g) -- (p) -- node[labelr] {\tiny \ndsep} (v)
					(v) -- node[labelal] {\tiny \ndsep} (x)
					(v) -- node[labelar] {\tiny \dsep} (y);
			\end{scope}
			\begin{scope}[local bounding box=zigzig2, xshift=21mm]
				\node (p) at (0,0) {$p$};
				\node[mainnode] (v) at (0,-1) {$v$}; \draw (p) -- (v);
				\node (g) at (2,-1) {$g$}; \draw (p) -- (g);
				\node (x) at (-1,-2) {$x$}; \draw (v) -- node[labelal] {\tiny \ndsep} (x);
				\node (y) at (1,-2) {$y$}; \draw (v) -- node[labelar] {\tiny \dsep} (y);
			\end{scope}
			\begin{scope}[local bounding box=zigzig3, xshift=44mm]
				\node[mainnode] (v) at (0,0) {$v$};
				\node (x) at (-1,-1) {$x$}; \draw (v) -- (x);
				\node (p) at (1,-1) {$p$}; \draw (v) -- (p);
				\node (y) at (0,-2) {$y$}; \draw (p) -- node[labelal] {\tiny \dsep} (y);
				\node (g) at (2,-2) {$g$}; \draw (p) -- (g);
			\end{scope}
			\draw[-latex] (zigzig1) -- node[above] {\rotlabel{p}} ($(zigzig2.north west)!(zigzig1)!(zigzig2.south west)$);
			\draw[-latex] ($(zigzig2.north east)!(zigzig1)!(zigzig2.south east)$) -- node[above] {\rotlabel{v}} ($(zigzig3.north west)!(zigzig1)!(zigzig3.south west)$);
		\end{tikzpicture}
		}
		\caption{ZIG-ZIG (rotate at $p$, then $v$)}\label{fig:zig-zig}
	\end{subfigure}
	\caption{Sketches of the two cases in a $\tsplayStep(v)$, including the behavior of all possible types of children of $v$. A small ``\sep'' (resp.\ ``\dsep'', ``\isep'') indicates an edge to a separator (resp.\ direct/indirect separator) child, and ``\nsep'' (resp.\ ``\ndsep'',  ``\nisep'') indicates the child cannot have the respective property.}\label{fig:splay-step}
\end{figure}

\Cref{fig:splay-step} sketches a ZIG-ZAG, resp.\ ZIG-ZIG step. In the figure, $x$, $y$, and $z$ represent different types of children that may or may not exist (and there may even be multiple non-separator children like $x$). In the first tree in \cref{fig:zig-zag}, observe that $\bd(T_z) \subseteq \{v\} \cup \bd(T_v)$, so $\bd(T_z) = \{v, g\}$. The remainder of the two sketches is easily explained from the definition of rotations; in particular, when rotating a node $v$ with its parent $p$, the direct separator child of $v$ becomes a child of $p$, and all other children of $v$ and $p$ are not affected.

We now present three Splay-based algorithms using the $\tsplayStep(v)$ procedure to bring a node to the root.\footnote{All three variants are found in \path{stt/src/twocut/splaytt.rs} in the source code.}

\paragraph{Greedy SplayTT.} Our first algorithm is similar to \algMRT{} and is inspired by the top tree implementation of Holm, Rotenberg, and Ryhl~\cite{HolmEtAl2023}. Greedy SplayTT brings a node $v$ to the root by repeatedly trying executing $\tsplayStep$ on $v$, its parent, and its grandparent.
See \cref{alg:greedy} for pseudocode.

\begin{figure}
	\centering
	\begin{minipage}{.7\textwidth}
		\begin{algorithmic}
			\Procedure{NodeToRootGreedySplay}{$v$}
				\While{$v$ has parent $p$}
					\If{$p$ has parent $g$}
						\If{$\tcanSplayStep(v)$}
							\State $\tsplayStep(v)$
						\ElsIf{$\tcanSplayStep(p)$}
							\State $\tsplayStep(p)$
						\Else 
							\State $\tsplayStep(g)$ \Comment{Must be possible}
						\EndIf
					\Else \Comment{$p$ is the root}
						\State $\trotate(v)$
					\EndIf
				\EndWhile
			\EndProcedure
		\end{algorithmic}
	\end{minipage}
	\caption{The Greedy SplayTT implementation of \tNodeToRoot{}.}\label{alg:greedy}
\end{figure}

The following lemma implies that Greedy SplayTT only performs valid rotations.

\begin{lemma}
	Let $v$ be a node in a 2-cut STT $T$ with parent $p$ and grandparent $g$. Then one of $\tsplayStep(v)$, $\tsplayStep(p)$, and $\tsplayStep(g)$ can be executed.
\end{lemma}
\begin{proof}
	Suppose all three calls are invalid. First, observe that then $g$ must have a parent $h$ and a grandparent~$i$; otherwise, $g$ is not a separator, so $\tsplayStep(v)$ is allowed. Now suppose $\tsplayStep(v)$, $\tsplayStep(p)$, $\tsplayStep(g)$ are all disallowed. Then $g$, $h$, and $i$ must be separators. But then $\tsplayStep(g)$ is allowed, a contradiction.
\end{proof}

It remains to show that Greedy SplayTT actually brings the given node to the root. For this, observe that performing $\tsplayStep(x)$ for some node $x$ decreases the depth of each child of $x$ by at least one (see \cref{fig:splay-step}). Hence, the depth of each \emph{descendant} of $x$ is also decreased. Each $\tsplayStep$ in $\Call{NodeToRootGreedySplay}{v}$ thus decreases the depth of $v$, eventually bringing it to the root.

\paragraph{Two-pass SplayTT.} We now describe the algorithm of Berendsohn and Kozma~\cite{BerendsohnKozma2022a}. Suppose we want to rotate $v$ to the root. The idea is to first do one pass over the root path of $v$ and remove all nodes that might inhibit rotations. Then we splay $v$ to the root. Notably, we do not use \cref{p:splay-step-allowed}, but the first pass ensures that \emph{every} rotation on the root path of $v$ is valid afterwards.

The algorithm uses the following helper procedure. Let $x$ be a descendant of a node $y$. The procedure $\tsplayTo(x,y)$ executes $\tsplayStep(x)$ until $y$ is the parent or grandparent of $x$. Then, if $y$ is the grandparent of $x$, it executes a final rotation, so $x$ becomes a child of $y$.

We now describe the algorithm. Consider a non-root node $x$ on the root path of $v$, and let $p$ be the parent of $x$. Recall a rotation at $x$ is \emph{not} allowed if and only if $\bdsize{T_x} = 1$ and $\bdsize{T_p} = 2$ (see \cref{p:rot-allowed}). In this case, we call $p$ a \emph{branching node}. We first find all branching nodes $b_1, b_2, \dots, b_k$ on the root path of $v$. We then call $\tsplayTo(v, b_1)$, and subsequently $\tsplayTo(b_i, b_{i+1})$ for each $i \in [k-1]$. Finally, we splay $b_k$ to the root by repeatedly calling $\tsplayStep(b_k)$. This concludes the first pass. The second pass simply consists of repeatedly calling $\tsplayStep$, until $v$ is the root.

It can be seen that only valid rotations are executed. We refer to Berendsohn and Kozma~\cite{BerendsohnKozma2022a} for more details.

\paragraph{Local Two-pass SplayTT.} We also implement a variant of Two-pass SplayTT that essentially does both passes at once. Bringing a node $v$ to the root works by repeating the following. If possible, we call $\tsplayStep(v)$. If not, then the parent or grandparent of $v$ must be a branching node, and we (essentially) perform a $\tsplayStep$ on it to bring it closer to the next higher branching node. Pseudocode for this variant is found in~\cref{alg:l2p}.

Note that this algorithm uses the condition of \cref{p:splay-step-allowed} to determine whether $\tsplayStep(v)$ can be executed. In some cases, this might ``skip'' a branching node that would have been handled separately by the non-local Two-pass SplayTT algorithm. Otherwise, the rotations executed in the Local variant are the same as in the non-local variant (just ordered differently), and correctness follows similarly.

\begin{figure}
	\centering
	\begin{minipage}{.7\textwidth}
		\begin{algorithmic}
			\Procedure{NodeToRootL2PSplay}{$v$}
				\While{$v$ has parent $p$}
					\If{$p$ has parent $g$}
						\If{$\tcanSplayStep(v)$}
							\State $\tsplayStep(v)$
						\Else \Comment{$g$ must be a separator}
							\LComment{Find a branching node}
							\If{$\tisSep(p)$}
								\LComment{$p$ is a branching node}
								\State $\tsplayStep(v)$ \Comment{Must be possible}
							\Else \Comment{$g$ is a branching node}
								\If{$\tcanSplayStep(g)$}
									\State $\tsplayStep(g)$
								\Else
									\State $\trotate(g)$
								\EndIf
							\EndIf
						\EndIf
					\Else \Comment{$p$ is the root}
						\State $\trotate(v)$
					\EndIf
				\EndWhile
			\EndProcedure
		\end{algorithmic}
	\end{minipage}
	\caption{The Local Two-pass SplayTT implementation of \tNodeToRoot{}.}\label{alg:l2p}
\end{figure}

\paragraph{Stability.} All four proposed algorithms (including MoveToRootTT) are stable, as we show in \cref{sec:stability}.

\paragraph{Running time.}\label{sec:splay-analysis}
We now turn to the running-time analysis of the SplayTT variants described above. Recall that we want to achieve $\fO(\log n)$ amortized time per call to $\tNodeToRoot{}$.

We use the potential method~\cite{Tarjan1985}. Our potential function is essentially the sum-of-logs function of Sleator and Tarjan~\cite{SleatorTarjan1985}. Let $T$ be an STT, and let $v \in V(T)$. We define $\phi(T_v) = c \cdot \log( |T_v|+1 )$, and $\Phi(T) = \sum_{v \in V(T)} \phi(T_v)$, where $c$ is a constant to be chosen later.

Since $\tsplayStep$ involves three nodes on a single path in the underlying tree, the following lemmas are easily generalized from the BST case.
\begin{lemma}\label{p:pot-rot}
	Let $T$ be an STT, and let $T'$ be produced from $T$ by a single rotation at $v \in V(T)$. Then $\Phi(T') - \Phi(T) \le 3( \phi(T'_v) - \phi(T_v) )$.
\end{lemma}
\begin{lemma}\label{p:pot-splay-step}
	Let $T$ be an STT, and let $T'$ be produced from $T$ by a ZIG-ZIG or ZIG-ZAG step at $v \in V(T)$. Then $\Phi(T') - \Phi(T) \le 3( \phi(T'_v) - \phi(T_v) ) - 2c$.
\end{lemma}

For Splay BST, the remainder of the analysis is easy. We keep executing splay steps at $v$. Thus, if all steps are ZIG-ZIG or ZIG-ZAG, and if $T^0, T^1, \dots, T^k$ is the sequence of trees produced, we have
\begin{align*}
	\Phi(T^k) - \Phi(T^0) & \le \sum_{i=0}^{k-1} 3 ( \phi(T^{i+1}_v) - \phi(T^i_v) ) - 2c\\
	& = 3 ( \phi(T^k_v) - \phi(T^0_v) ) - 2ck\\
	& \le 3 \log |V(T)| - 2ck.
\end{align*}

Setting $c$ to the running time of a single rotation yields the desired amortized running time $\fO( \log |V(T)| )$. The very last step might be a ZIG, but then the amortized running time increases only by an additive constant.

In our SplayTT variants, the splay steps do not produce a single telescoping sum as above. However, we can split the splay steps into a constant number of sets that do telescope nicely.

For Two-pass SplayTT, we refer to Berendsohn and Kozma's analysis~\cite{BerendsohnKozma2022a}. Essentially, each pass produces a telescoping sum adding up to $\fO(\log n)$. While they use a different potential function (which is necessary to prove static optimality), replacing their analysis of $\tsplayStep$ with \cref{p:pot-rot,p:pot-splay-step} yields an overall $\fO( \log n )$ amortized running time. Local Two-pass SplayTT only skips the removal of some branching nodes, which does not affect analysis.

For Greedy SplayTT, we use a different analysis.

\begin{lemma}\label{p:greedy-running-time}
	Performing $\tNodeToRoot{}$ with greedy SplayTT has amortized running time $\fO(\log n )$, where $n = |V(T)|$.
\end{lemma}
\begin{proof}
	Fix $v$ and consider a call $\tNodeToRoot{v}$ with greedy SplayTT. Let $T$ be the current tree before or after some $\tsplayStep$ during execution, let $p$ and $g$ be the parent and grandparent of $v$ in $T$ (either is $\bot$ if that node does not exist). For $x \in V(T)$, define $\psi(T_x) = 3 \cdot \phi(T_x)$, and we define $\psi(T_x) = 3c \cdot \log (n+1)$ for non-existing nodes $x = \bot$. Let $\Psi(T) = \psi(T_v) + \psi(T_p) + \psi(T_g)$. We claim that $\Psi(T)$ is an upper bound for the amortized running time so far, for (essentially) every intermediate tree $T$, and in particular the final tree. This clearly implies the lemma.
	
	At the start, the claim is trivially true. Now suppose we execute some $\tsplayStep$. Let $T$, $T'$ be the STT before and after the execution, and let $p$, $p'$, $g$, $g'$ be the parent and grandparent of $v$ in the respective tree. Recall that the amortized running time of $\tsplayStep(x)$ is $3(\phi(T'_x) - \phi(T_x)) = \psi(T'_x) - \psi(T_x)$.
	\begin{enumerate}[(i)]
		\item\label{item:greedy-analysis-simple} If we call $\tsplayStep(v)$, then the amortized cost is $\psi(T'_v) - \psi(T_v)$. The nodes $p$ and $g$ are simply removed from the root path of $v$, so $p'$ and $g'$ are ancestors of $g$ in $T$ (or nonexistent). This implies that $\psi(T_p) \le \psi(T'_{p'})$ and $\psi(T_g) \le \psi(T'_{g'})$, so
		\begin{align*}
			\psi(T'_v) - \psi(T_v) \le \Psi(T') - \Psi(T).
		\end{align*}
		
		\item\label{item:greedy-analysis-complicated} Suppose we call $\tsplayStep(p)$. Since $\tsplayStep(v)$ was disallowed, we know (by \cref{p:splay-step-allowed}) that at least one of $v$ and $p$ is not a separator in $T$, and $g$ is a separator in $T$.
		
		If $p$ is a separator in $T$, then $v$ is not. This means that $\tsplayStep(p)$ removes $g$ and its parent $h$ from the root path of $v$ (i.e., $v$ stays a child of $p$), and essentially the same analysis as in \ref{item:greedy-analysis-simple} applies.
		
		If $p$ is not a separator, then a ZIG-ZIG step at $p$ is performed. If $v$ is a not a direct separator, again $g$ and $h$ are removed from the root path of $v$.
		Otherwise, we cannot guarantee that the invariant holds after this step. However, we can show that the step directly after is a ZIG-ZAG step, and the invariant holds afterwards.
		
		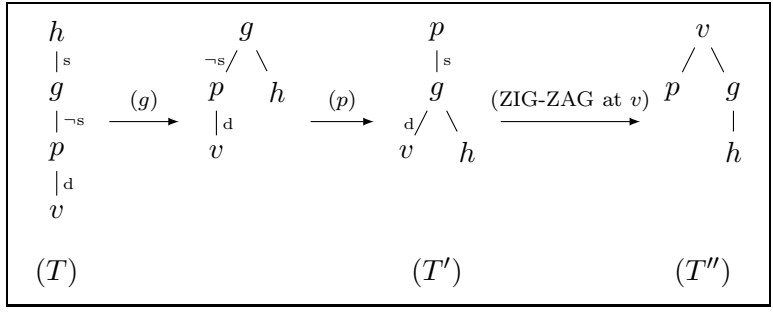
\begin{figure}
			\centering
			\fbox{
				\begin{tikzpicture}[x=4mm, y=8mm, rotarrow/.style={-latex, shorten >=2mm, shorten <=2mm}]
					\begin{scope}[local bounding box=stage1]
						\node (h) at (0,0) {$h$};
						\node (g) at (0,-1) {$g$}; \draw (h) -- node[labelr] {\tiny \sep} (g);
						\node (p) at (0,-2) {$p$}; \draw (g) -- node[labelr] {\tiny \nsep} (p);
						\node (v) at (0,-3) {$v$}; \draw (p) -- node[labelar] {\tiny \dsep} (v);
						\node at (0,-4) {$(T)$};
					\end{scope}
					\begin{scope}[local bounding box=stage2, xshift=25mm]
						\node (g) at (0,0) {$g$};
						\node (p) at (-1,-1) {$p$}; \draw (g) -- node[labelal] {\tiny \nsep} (p);
						\node (h) at (1,-1) {$h$}; \draw (g) -- (h);
						\node (v) at (-1,-2) {$v$}; \draw (p) -- node[labelar] {\tiny \dsep} (v);
					\end{scope}
					\begin{scope}[local bounding box=stage3, xshift=50mm]
						\node (p) at (0,0) {$p$};
						\node (g) at (0,-1) {$g$}; \draw (p) -- node[labelr] {\tiny \sep} (g);
						\node (v) at (-1,-2) {$v$}; \draw (g) -- node[labelal] {\tiny \dsep} (v);
						\node (h) at (1,-2) {$h$}; \draw (g) -- (h);
						\node at (0,-4) {$(T')$};
					\end{scope}
					\begin{scope}[local bounding box=stage5, xshift=85mm]
						\node (v) at (0,0) {$v$};
						\node (p) at (-1,-1) {$p$}; \draw(v) -- (p);
						\node (g) at (1,-1) {$g$}; \draw (v) -- (g);
						\node (h) at (1,-2) {$h$}; \draw (g) -- (h);
						\node at (0,-4) {$(T'')$};
					\end{scope}
					\node (arrhelp) at (0,-1.5) {};
					\newcommand{\stagearrow}[3]{\draw[rotarrow] ($(#1.north east)!(arrhelp)!(#1.south east)$) -- node[above] {#2} ($(#3.north west)!(arrhelp)!(#3.south west)$)}
					\stagearrow{stage1}{\rotlabel{g}}{stage2};
					\stagearrow{stage2}{\rotlabel{p}}{stage3};
					\stagearrow{stage3}{\scriptsize (ZIG-ZAG at $v$)}{stage5};
				\end{tikzpicture}
			}
			\caption{Illustration of the special case in \cref{p:greedy-running-time} \ref{item:greedy-analysis-complicated}, where we first perform a ZIG-ZIG at $p$, and then a ZIG-ZAG at $v$. ``\sep'', ``\nsep'', and ``\dsep'' have the same meaning as in \cref{fig:splay-step}.}\label{fig:greedy-special-case}
		\end{figure}

		\Cref{fig:greedy-special-case} illustrates the situation. The key insight is that $g$ is a separator in $T'$ (after the current ZIG-ZIG step at $p$). To see this, observe that $\bd(T_g) = \{h, a\}$ for some node $a$ that is a proper ancestor of $h$. On the other hand, we have $a \notin \bd(T_p)$, since $p$ is not a separator. Further observe that $V(T'_g) \supseteq V(T_g) \setminus V(T_p)$. Thus, we have $a \in \bd(T'_g)$, so $g$ is a separator in $T'$.
		
		Since $v$ is still a direct separator in $T'$ (note that $\bd(T_v) = \bd(T'_v) = \{p,g\}$ by assumption), the next $\tsplayStep$ will be a ZIG-ZAG at $v$. Call the resulting tree $T''$ (see \cref{fig:greedy-special-case}). We now prove that $\Psi(T'') - \Psi(T)$ bounds the amortized running time $t$ of the two $\tsplayStep$s, which is
		\[t = \psi(T'_p) - \psi(T_p) + \psi(T''_v) - \psi(T'_v).\]
		
		We have $V(T_v) = V(T'_v)$, $V(T'_p) = V(T''_v)$, implying
		\begin{align*}
			& \psi(T''_v) - \psi(T'_v) = \psi(T''_v) - \psi(T_v), & \text{and}\\
			& \psi(T'_p) - \psi(T_p) = \psi(T''_v) - \psi(T_p) \le \psi(T''_{p''}) - \psi(T_p).
		\end{align*}
		
		Moreover, since $V(T_h) = V(T''_v)$, we have \[ \psi(T_g) < \psi(T_h) = \psi(T''_v) \le \psi(T''_{g''}), \]
		so $\psi(T''_{g''}) - \psi(T_g) > 0$. Thus,
		\[ t < \psi(T''_v) - \psi(T_v) + \psi(T''_{p''}) - \psi(T_p) + \psi(T''_{g''}) - \psi(T_g) = \Psi(T'') - \Psi(T), \]
		so the running time of both $\tsplayStep$s together is bound by the change in~$\Psi$.
		
		\item Finally, suppose we call $\tsplayStep(g)$. Then $\tsplayStep$ at both $v$ and $p$ must be disallowed. The former implies that $g$ is a separator, and the latter implies that $p$ and $g$ cannot be separators at the same time, so $p$ is not a separator (in $T$). Thus, $\tsplayStep(g)$ simply removes $h$ and the parent of $h$ from the root path of $v$, so we have $g = g'$ and $p = p'$ and the potential of $v$ and $p$ does not change. Hence, the amortized cost $\psi(T'_g) -\psi(T_g)$ precisely matches the change $\Psi(T')-\Psi(T)$.\qedhere
	\end{enumerate}
\end{proof}

It remains to analyze potential increases in \opLink{}, \opCut{}, and \opPathWeight{}, which we do below.
\begin{theorem}
	Starting with a dynamic forest on $n$ nodes without edges, using any of the above SplayTT variants, $m$ operations \opLink{}, \opCut{}, and/or \opPathWeight{} are performed in time $\fO( n + m \log n )$.
\end{theorem}
\begin{proof}
	\opLink{}, \opCut{}, and \opPathWeight{} each require up to two calls to \tNodeToRoot{}, along with a constant amount of additional work, for an amortized cost of $\fO( \log n )$. \opLink{} and \opCut{} additionally change the tree structure at the end, which changes the potential. $\opCut{u,v}$ only decreases the potential of $v$, since $v$ loses a child. $\opLink{u,v}$ adds $u$ as a child of $v$ at the end. However, at that point, $v$ is the root of its tree, hence only the potential of $v$ increases, and by at most $3 c \cdot \log n$. Thus, the amortized time of every operation is $\fO( \log n )$. At the start, the forest has no edges, so search trees consist of only one node each, for a total starting potential of $\fO(n)$. This yields at total running time of $\fO( n + m \log n )$.
\end{proof}

\subsubsection{SplayTT vs.\ link-cut trees}\label{sec:splaytt-vs-link-cut}

Sleator and Tarjan's link-cut trees~\cite{SleatorTarjan1983,SleatorTarjan1985} maintain a decomposition of the underlying tree into paths, and consist of a hierarchy of BSTs on these paths. Moving a node $v$ to the root is performed roughly as follows. First, for each BST $B$ between $v$ and the root, splay the node $v_B$ to the root of $B$, where $v_B$ is the lowest node in $B$ that is an ancestor of $v$ in the overall link-cut tree. This shortens the path from $v$ to the root, such that every node on that path comes from a different BST. Then, an operation called \emph{splice} is performed, which splits and merges BSTs until the path from $v$ to the root is contained in a single BST. Finally, $v$ is splayed to the root.

The Two-pass SplayTT algorithm, in a way, works very similar. Bose, Cardinal, Iacono, Koumoutsos, and Langerman~\cite{BoseCardinalEtAl2020} observed that link-cut trees essentially are 2-cut STTs. We remark that, disregarding the left-right order of BST nodes in the link-cut tree, SplayTT performs almost the same rotations as link-cut trees. The main difference is that no path-decomposition is maintained; instead, SplayTT ``automatically'' detects a decomposition of the search path.

As discussed in \cref{sec:intro:rooted-unrooted}, link-cut trees ``natively'' maintain \emph{rooted} forests. Each path in the decomposition can be seen as oriented towards the root, and this orientation is preserved by the left-right order in the corresponding BST. To support \opEvert{} (and hence arbitrary \opLink{}s in unrooted forests), it must be possible to \emph{reverse} paths and the corresponding BSTs. To preserve the $\fO(\log n)$ amortized cost this has to be done lazily using a \emph{reverse bit}. This complicates the implementation, and we speculate that it is the main reason why our data structures outperform link-cut trees in our experiments with unrooted forests.

\section{Experimental evaluation}\label{sec:experiments}

We performed multiple experiments on our Rust implementation. We used an AMD Ryzen 5 2600X processor running Debian \emph{Bullseye} at 3.6~GHz with 512~KB of L1 cache, 3~MB of L2 cache, and 16~MB of L3 cache.

\subsection{Algorithms}\label{sec:exp-algos}
We now describe the data structures we implemented and compared in our experiments.

\paragraph{Edge-weighted unrooted forests.}
Using our basic 2-cut STT data structure (\cref{sec:2cut-impl}) and the \opLink{}, \opCut{}, and \opPathWeight{} procedures described in \cref{sec:link-cut,sec:edge-weights} together with one of the four \tNodeToRoot{} algorithms (Greedy SplayTT, Two-pass SplayTT, Local two-pass SplayTT, and MoveToRootTT), we obtain four different implementations. We call these implementations \emph{Stable} since they assume stability.

As mentioned in \cref{sec:link-cut}, our Rust implementation also includes procedures for \opLink{}, \opCut{}, and \opPathWeight{} that do not assume stability, but additionally require the \tNodeBelowRoot{} operation. We have four corresponding \tNodeBelowRoot{} implementations as slight variants of the four \tNodeToRoot{} implementations.\footnote{Found next to the respective \tNodeToRoot{} implementation in the source code.}

We denote the resulting eight data structures as (\dsimpl{Stable}) \dsimpl{Greedy Splay}, (\dsimpl{Stable}) \dsimpl{2P Splay}, (\dsimpl{Stable}) \dsimpl{L2P Splay}, and (\dsimpl{Stable}) \dsimpl{MTR}.

Further, we have \dsimpl{Link-cut}\footnote{Found in \implfile{stt/src/link_cut.rs} in the source code.}, an implementation of the amortized variant of Sleator and Tarjan's link-cut trees~\cite{SleatorTarjan1985}, where the handling of edge weights is similar to the way described in \cref{sec:edge-weights}.\footnote{Sleator and Tarjan only describe how to maintain vertex weights. Tarjan and Werneck~\cite{Werneck2006} simulate edge weights by adding a vertex on each edge and maintaining vertex weights. We did not test this approach.}

Finally, we have two linear-time data structures. \dsimpl{1-cut}\footnote{Found in \implfile{stt/src/onecut.rs}} is a naive dynamic forest implementation that maintains a rooting of each tree (i.e., a 1-cut search tree on each tree). The dynamic forest operations are implemented as described above, where $\tNodeToRoot{v}$ repeatedly rotates the root with one of its children until $v$ is the root.
\dsimpl{Petgraph}\footnote{Found in \implfile{stt/src/pg.rs}} is a naive dynamic forest implementation using the Petgraph\footnote{\url{https://crates.io/crates/petgraph}} library, which appears to be the most popular graph library for Rust at the time of writing. We tested the other implementations using \dsimpl{Petgraph} as a reference.

\paragraph{Rooted forests.}
We also implemented data structures maintaining rooted forests without edge weights. We support \opLink{}, \opCut{}, \opFindRoot{} and (depending on the experiment) \opEvert{}.
An extension of our STT-based data structures is sketched in \cref{sec:stts-rooted} and yields four variants \dsimpl{Greedy Splay}, \dsimpl{2P Splay}, \dsimpl{L2P Splay}, and \dsimpl{MTR}. Since we use both \tNodeToRoot{} and \tNodeBelowRoot{}, there are no stable variants.

\dsimpl{Link-cut} is the same link-cut tree implementation as above. If \opEvert{} is not needed, we disable any checks and modifications of the reverse bit (though the slight space overhead remains).
Finally, \dsimpl{Simple}\footnote{Found in \implfile{stt/src/rooted.rs}} is a naive implementation that maintains the rooted forest explicitly via parent pointers.

\subsection{Experiments and results}
We now describe our experiments and discuss their results. To reduce variance, we performed every (sub)experiment between ten and twenty times. This section only shows a selection of results. More detailed result tables are found in \cref{sec:detailed-results}. All experiments can be reproduced by calling a single script provided with the source code; see the included \texttt{README.md} file for more details.

\newcommand{\dsMTRorStable}{(\dsimpl{Stable}) \dsimpl{MTR}}

\paragraph{Uniformly random connectivity queries.}\label{sec:exp:uniform}
In our first experiment, weights are empty, so the updating logic from \cref{sec:edge-weights} is not required and $\opPathWeight{u,v}$ simply returns whether $u$ and $v$ are connected or not. This allows us to directly compare the dynamic forest implementations without edge weight handling.

A list of queries is pre-generated, starting with an empty forest. For each query, we draw two vertices $u,v$ uniformly at random; if $u$ and $v$ are not connected, we call $\opLink{u,v}$; otherwise, we either call $\opPathWeight{u,v}$ or call $\opCut{}$ on some edge on the path between $u$ and~$v$, with probability $\frac{1}{2}$ each. We then execute the list of queries once for each implementation.

First, we compared all implementations for $n \le 1000$ vertices and $m = 20n$ queries. \dsimpl{Petgraph} performed very badly (worse than the next-worst implementation by a factor of over 15 at $n = 1000$), so we excluded it from all further experiments. We then tried larger values of $n$ up to 8000, with $m = 100n$.\footnote{The maximum value for $n$ is chosen such that the asymptotically worse behavior of \dsimpl{1-cut} is clearly visible, but the overall experiment still takes a reasonable amount of time. The same applies to the other experiments.}
The results for $n = 8000$ are shown in the first column of \cref{tab:urc-msf-fdc} (see \cref{tab:full_queries} in \cref{sec:detailed-results} for more details).

Our SplayTT variants consistently outperform Link-cut trees by up to 25\%. Among them, the stable variants are usually slightly faster than the non-stable ones, and \dsimpl{Greedy Splay}/\dsimpl{L2P Splay} are slightly faster than \dsimpl{2P Splay}. All of this points towards simple algorithms performing better in practice.

The even simpler \dsimpl{MTR} and \dsimpl{Stable MTR} are faster than all Splay-based data structures, perhaps because of the uniformity of the input (as discussed in \cref{sec:mtr}). The simple linear-time 1-cut data structure is faster for smaller values of $n$ (up to around 3000, see \cref{tab:full_queries} in \cref{sec:detailed-results}), but is the worst by some margin at $n = 8000$.

\paragraph{Incremental MSF.}\label{sec:exp:inc-msf}
Our second, more practical experiment consists of solving the incremental minimum spanning forest (MSF) problem.

We are given the edges of a weighted graph one-by-one and have to maintain an MSF. Edges are never removed. A simple solution using dynamic forests works as follows. Whenever an edge $\{u,v\}$ with weight~$w$ arrives, if $u$ and $v$ are in different components, add the edge to the forest. Otherwise, find the heaviest edge on the path from $u$ to $v$, and if its weight is larger than $w$, replace it with the new edge.

To find the actual heaviest edge instead of just its weight, we extend our edge weight monoid $(\N, \max)$ to also contain a heaviest edge. The result is still a monoid, hence our algorithms can be used without change.

\begin{table}[t]
	\centering\small
	\begin{tabular}{lccc}
		\toprule
		& \multicolumn{3}{c}{Time (\textmu{}s/query or \textmu{}s/edge)}\\
		\cmidrule{2-4}
		Algorithm & URC & \multicolumn{2}{c}{MSF}\\
		& $n = 8000$ & $n = 10^6$ & \texttt{ogbl}\\
		\midrule
		Link-cut & $0.65$ & $3.43$ & $0.58$\\
Greedy Splay & $0.54$ & $2.94$ & $0.39$\\
Stable Greedy Splay & $0.52$ & $2.90$ & $0.40$\\
2P Splay & $0.56$ & $3.05$ & $0.40$\\
Stable 2P Splay & $0.53$ & $2.99$ & $0.40$\\
L2P Splay & $0.54$ & $3.01$ & $0.39$\\
Stable L2P Splay & $0.52$ & $3.01$ & $0.40$\\
MTR & $0.44$ & $2.59$ & $0.36$\\
Stable MTR & $0.44$ & $2.60$ & $0.39$\\
1-cut & $0.77$ & $5.32$ & $0.18$\\
Kruskal (petgraph) & -- & $0.63$ & $0.17$\\

		\bottomrule
	\end{tabular}
	\caption{Results for uniformly random connectivity (URC), incremental MSF for random queries, and incremental MSF on the \texttt{ogbl-collab} dataset.}\label{tab:urc-msf-fdc}
\end{table}

As a first experiment, we follow Tarjan and Werneck~\cite{TarjanWerneck2010} and randomly generate inputs on $n \le 10^6$ vertices with $m = 8n$ edges.
Second, we use the \texttt{ogbl-collab} data set\footnote{Available under the ODC Attribution License at\\\url{https://ogb.stanford.edu/docs/linkprop/\#ogbl-collab}}~\cite{HuFeyEtAl2020} to generate an input that might be closer to real-world applications. The data set consists of a set of authors and collaborations between authors, annotated with a year. We interpret this as a dynamically changing graph where the first collaboration creates an edge with weight 1, and each subsequent collaboration increases the weight of the edge. Inverting the edge weights yields a natural dynamic MSF problem, with the additional allowed operation of \emph{decreasing} an edge weight, which can be easily implemented by first removing the edge (if it exists in the current MSF), and then adding it again with the new weight. The resulting input has 235\,868 vertices and 1\,179\,052 queries.

We also compare the online algorithms with the Petgraph library's implementation of Kruskal's offline algorithm.

Kruskal's algorithm outperforms the online algorithms by a large factor (this is expected, since it is offline and less general). Otherwise, the results of this experiment are similar to the uniformly random query experiments, except that \dsimpl{Stable Greedy Splay} is now clearly the fastest among the Splay-based data structures. It is not clear why this is not the case in the previous experiment, but we note that \dsimpl{Stable Greedy Splay} is our simplest SplayTT-based data structure. The results of the \texttt{ogbl-collab} experiment are similar except that the the difference between stable and non-stable implementations does not exist, for unknown reasons.

\paragraph{Random queries with variable probability of \opPathWeight{}.} Informal experiments with a naive fully-dynamic connectivity algorithm lead us to believe that \dsimpl{Link-Cut} performs better compared to our approaches when \opPathWeight{} queries are common (and thus the reverse bit is rarely changed). Hence, we repeated the first experiment (\cref{sec:exp:uniform}) with $n = 5000$, except that the probability $p$ of generating a \opPathWeight{} query (instead of a \opCut{}) is variable. \Cref{fig:query-path-prob} shows that the performance of link-cut gets closer to the STT data structures as $p$ approaches 1 (even slightly outperforming the weaker \dsimpl{2P~Splay} variant), confirming our suspicion.

\paragraph{Degenerate queries.}
\dsimpl{MTR} and \dsimpl{Stable MTR} outperform the other algorithms on uniform queries, despite having asymptotic worst-case performance of $\Theta(n)$ per operation. To experimentally confirm the worst-case behavior, we create a path $G$ of $n \le 10\,000$ nodes $v_1, v_2, \dots, v_n$, and then call $\opPathWeight{v_i, v_n}$ for all $i \in [n]$ in order. While the queries have strong locality, the two vertices $v_i, v_n$ are very far from each other on average. All Splay-based approaches are able to exploit the locality and outperform the linear-time data structures (\dsimpl{MTR}, \dsimpl{Stable MTR}, and \dsimpl{1-cut}) by a factor of over 100 when $n = 10\,000$.

To check how ``isolated'' our degenerate example is, we performed the following ``noisy'' experiment. Fixing $n = 5000$, for each $i \in [n]$, we call $\opPathWeight{v_j, v_n}$, where $j = i + \lfloor x \rfloor$ and $x$ is drawn from a normal distribution with mean 0 and standard deviation $\sigma$, for some values $\sigma \le 300$.
(See \cref{fig:degenerate-noisy}.)
As expected, \dsimpl{1-cut} still performs very badly, since the added noise does not change the expected distance between $v_i$ and $v_n$. \dsimpl{MTR} and \dsimpl{Stable MTR}, on the other hand, do adapt, though even with $\sigma = 300$ both are still slower than the Splay-based variants by at least 10\%.

\begin{figure}
	\newcommand\plotRed{\addplot[color=red, mark=*, mark options={scale=0.6}]}
	\newcommand\plotGreen{\addplot[color={black!40!green}, mark=square*, mark options={scale=0.5}]}
	\newcommand\plotBlue{\addplot[color=blue, mark=diamond*, mark options={scale=0.7}]}
	
	\begin{minipage}[c]{.5\textwidth}
		\centering
		\begin{tikzpicture}
			\begin{axis}[
				width={0.8\textwidth},
				height={50mm},
				xlabel={$p$},
				ylabel={\textmu{}s/query},
				xmin=0, xmax=1,
				ymin=0, ymax=0.8,
				legend pos=south east
			]
				\pgfplotstableread{figures/data_queries_path_prob.txt}\datatable
				
				\plotBlue table[y = Link-cut] from \datatable;
				\addlegendentry{\small Link-cut}
				
				\plotRed table[y = 2P_Splay] from \datatable;
				\addlegendentry{\small 2P Splay}
				
				\plotGreen table[y = Stable_Greedy_Splay] from \datatable;
				\addlegendentry{\small Stable Greedy Splay}
			\end{axis}
		\end{tikzpicture}
	\end{minipage}%
	\begin{minipage}[c]{.5\textwidth}
		\centering
		\begin{tikzpicture}
			\begin{axis}[
				width={0.8\textwidth},
				height={50mm},
				xlabel={$\sigma$},
				ylabel={\textmu{}s/query},
				xmin=0, xmax=300,
				ymin=0, ymax=25,
				legend pos=north east
			]
				\pgfplotstableread{figures/data_degenerate_noisy.txt}\datatable
				
				\plotRed table[y = MTR] from \datatable;
				\addlegendentry{\small MTR}
				
				\plotGreen table[y = 1-cut] from \datatable;
				\addlegendentry{\small 1-cut}
				
				\plotBlue table[y = Link-cut] from \datatable;
				\addlegendentry{\small Link-cut}
			\end{axis}
		\end{tikzpicture}
	\end{minipage}
	
	\begin{minipage}[t]{.45\textwidth}
		\captionof{figure}{Random queries with increasing probability $p$ of \opPathWeight{}.}\label{fig:query-path-prob}
	\end{minipage}\hfill%
	\begin{minipage}[t]{.45\textwidth}
		\captionof{figure}{Noisy degenerate input.}\label{fig:degenerate-noisy}
	\end{minipage}
\end{figure}

\paragraph{Lowest common ancestors.}
In our final two experiments, we maintain a \emph{rooted} forest on $n$ vertices and execute $10n$ queries among $\opLink{u, v}$, $\opCut{v}$, and $\opLCA{u,v}$, the latter of which returns the lowest common ancestor of two nodes in the same tree. The query distribution is as follows. A random non-root node is \opCut{} with probability
$\frac{1}{2} \cdot \frac{m}{n-1}$, where $m$ is the current number of non-root nodes. Otherwise, a pair of nodes $\{u,v\}$ is generated uniformly at random, and $\opLink{u,v}$ or $\opLCA{u,v}$ is chosen depending on whether $u$ and $v$ are in the same tree. Overall, we have roughly 46\% \opLink{}s, 38\% \opCut{}s and 16\% \opLCA{}s.

In the second experiment, we additionally allow $\opEvert{v}$, i.e., changing the root of a tree. Each $\opCut{}$ is replaced with $\opEvert{}$ with probability $\frac{1}{2}$, resulting in roughly 30\% \opLink{}s, 20\% \opCut{}s, 30\% \opLCA{}s and 20\% \opEvert{}s.

As expected, \dsimpl{Link-cut} outperforms our data structures considerably in the first experiment, where only the latter have to maintain extra data (to represent rooted trees). When \opEvert{} is allowed, somewhat surprisingly, STT-based data structures are faster again. The \dsimpl{Simple} data structure performed much worse than all others and was excluded from experiments with large~$n$. See \cref{tab:full_lca,tab:full_lca_evert} for more details.

\begin{table}
	\centering\small
	\begin{tabular}{lccc}
		\toprule
		& \multicolumn{3}{c}{Time (\textmu{}s/query)}\\
		\cmidrule{2-4}
		Algorithm & \multicolumn{2}{c}{Without \opEvert{}} & With \opEvert{}\\
		& $n = 20000$ & $n = 10^6$ & $n = 10^6$\\
		\midrule
		Link-cut & $0.31$ & $0.83$ & $1.34$\\
Greedy Splay & $0.42$ & $0.94$ & $1.12$\\
2P Splay & $0.43$ & $0.95$ & $1.13$\\
L2P Splay & $0.41$ & $0.93$ & $1.13$\\
MTR & $0.34$ & $0.85$ & $1.00$\\
Simple & $2.59$ & -- & --\\

		\bottomrule
	\end{tabular}
	\caption{Results for the LCA experiment.}\label{tab:lca}
\end{table}

\subsection{Notes on the Rust implementation}

All implementations share a common interface (the Rust \emph{traits} \texttt{DynamicForest}, resp.\ \texttt{RootedDynamicForest}) that is used by the experiments. Code is reused whenever possible through heavy use of generics.

There are some differences between the pseudocode presented here and the actual Rust implementation. This is due to the fact that procedures like $\tcanRotate(v)$ and $\tcanSplayStep(v)$ contain multiple $\tisSep(\cdot)$ checks, which can cause an unnecessarily large number of calls to the $\tparent(\cdot)$ function, even though the parent and possibly further ancestors of $v$ may be already known (consider, e.g., \cref{alg:mtr,alg:greedy}). Hence, we eliminated some of the additional calls by, e.g., introducing a function $\mathtt{is\_separator\_hint}(v,p)$, which is more efficient, but requires $p = \tparent(v)$ to be given.

We applied this principle liberally in all STT-based variants and our \dsimpl{Link-cut} implementation. The performance gains were small, except for \dsimpl{Greedy Splay}, which was slightly slower than the other STT-based variants before, and now is slightly faster. We did not attempt any fine-tuning beyond this.

\section{Conclusion}\label{sec:conclusion}

We presented a new framework to implement dynamic forests based on STTs. Our data structures are as capable as link-cut trees, with a wide range of applications.
For maintaining unrooted forests, our framework is arguably conceptually simpler than link-cut trees, since there is no need to explicitly maintain a (directed) path decomposition of the underlying forest. The main complexity lies in the implementation of the STT rotation primitive, which is easily separated and reused, simplifying the engineering of new variants. In contrast, variants of link-cut trees are somewhat restricted by the decomposition into BSTs; for example, no equivalent of our Greedy SplayTT algorithm for link-cut trees exists.

In our experiments, the SplayTT-based data structures outperform link-cut trees by 15-20\% if the dynamic forest is unrooted. Link-cut trees in turn are roughly 15-25\% faster for rooted dynamic forests (without the root-changing \opEvert{} operation).
A next step would be to attempt fine-tuning of our implementations and compare them with existing dynamic forest (in particular link-cut tree) implementations.

Among the SplayTT-based variants we tested, \dsimpl{Stable Greedy Splay} generally performed best, and is also the simplest to implement and analyze. However, the even simpler \dsimpl{MTR} algorithm outperformed our more sophisticated algorithms by around 15\%, except for specifically constructed inputs. It would be interesting to investigate whether there exist practical applications where the adaptivity of Splay-based data structures makes up for their increased complexity.

\bibliography{info}

\clearpage

\appendix

\section{Unique representation of the underlying tree}\label{sec:represent-proof}

In this section, we show that the representation of STTs presented in \cref{sec:2cut-impl} is sufficient to uniquely represent the underlying tree. Recall that we store the pointers $\tparent$, $\tdsepchild$, and $\tisepchild$. The parent pointers tell us the structure of the tree, and the child pointers tell us precisely which nodes are direct or indirect separators. We first show that this uniquely determines the boundaries of subtrees.

\begin{lemma}
	Given an STT $T$ and the pointers $\tparent$, $\tdsepchild$, $\tisepchild$ for each node, we can determine $\bd(T_v)$ for each node $v$.
\end{lemma}
\begin{proof}
	For the root $r$, we always have $\bd(T_r) = \emptyset$. If $v$ is not a separator and not the root, then $\bd(T_v)$ contains only the parent of $v$. If $v$ is a direct separator, then $\bd(T_v)$ consists of the parent and grandparent of $v$.
	
	Now consider an indirect separator node $v$ with parent $p$ and grandparent $g$. \Cref{p:boundary-rules} implies that $\bd(T_v) = \{p,x\}$, where $x \in \bd(T_p)$. Since $v$ is an indirect separator node, $x \neq g$. But $g \in \bd(T_p)$, so $x$ must be the remaining node in $\bd(T_p) \setminus g$. This observation allows us to determine all subtree boundaries in a top-down fashion.
\end{proof}

Once we have determined subtree boundaries, we can determine the edges of the underlying tree using the following lemma.

\begin{lemma}
	Let $T$ be a search tree on a tree $G$. Let $u, v \in V(T)$ such that $u$ is an ancestor of $v$. Then $\{u,v\} \in E(G)$ if and only if $u \in \bd(T_v)$, but $u \notin \bd(T_c)$ for each child $c$ of $v$.
\end{lemma}
\begin{proof}
	If $u \notin \bd(T_v)$, then there is no edge in $G$ between $u$ and $V(T_v)$, so, in particular, $\{u,v\} \notin E(G)$. Now suppose $u \in \bd(T_v)$. Since $G$ is a tree and $G[V(T_v)]$ is connected, there must be exactly one edge $\{u,x\}$ between $u$ and $V(T_v)$. If $u \in \bd(T_c)$ for some child $c$ of $v$, then $x \in V(T_c)$, so $\{u,v\} \notin E(G)$. Otherwise, we have $x \notin V(T-c)$ for each child $c$ of $v$, and hence $x = v$.
\end{proof}

\section{Implementing STT rotations}\label{sec:rot-impl}

\begin{lemma}\label{p:impl-rot}
	Given a node $v$ in an STT $T$, represented as described in \cref{sec:2cut-impl}, we can rotate $v$ with its parent in $\fO(1)$ time.
\end{lemma}
\begin{proof}
	Let $p = \tparent(v)$, let $g = \tparent(p)$, and let $c = \tdsepchild(v)$ ($g$ and/or $c$ may be $\bot$). We denote by $T'$ the tree after the rotation, and by $\tparent'(\cdot)$, $\tdsepchild'(\cdot)$, $\tisepchild'(\cdot)$ the correct respective pointers in $T'$. In the following, we frequently make use of the fact that each node has at most one direct and at most one indirect separator child (\cref{p:sep-nodes}).
	
	For the parent pointers, we have $\tparent'(v) = g$, $\tparent'(p) = v$, and, if $c \neq \bot$, additionally $\tparent'(c) = p$.
	
	If $g \neq \bot$, we may need to adjust its child pointers. Observe that $V(T'_v) = V(T_p)$, so $\bd(T'_v) = \bd(T_p)$. Thus, $v$ is an (in)direct separator child in $T'$ if and only if $p$ was an (in)direct separator child in $T'$. One of $\tdsepchild'(g)$ and $\tisepchild'(g)$ may accordingly change from $p$ to $v$.
	
	We now consider the child pointers of $p$. Note that $p$ gains a new parent ($v$) and keeps all other ancestors, loses a child ($v$) and possibly gains a child ($c$). First, we have $\tdsepchild'(p) = c$, since $c$ is the unique node with $\bd(T_c) = \{v,p\}$ if such a node exists; otherwise, $\tdsepchild'(p) = \bot = c$.
	
	If $p$ has a separator child $x \neq v$ in $T$, then $x$ clearly still is a separator in $T'$, and $v \notin \bd(T_x) = \bd(T'_x)$, so $\tisepchild'(p) = x$. If $x$ does not exist, then $p$ has no separator child in $T$. Since $p$ does not gain children in $T'$ besides $c$, this means $\tisepchild'(p) = \bot$.
	
	\begin{figure}[t]
		\centering
		\begin{tikzpicture}[xscale=0.7]
		\node at (0,-1) {};
		\node at (0,1) {};
		
		\node[vertex] (p) at (0,0) {};
		\node[vertex] (v) at (1,0) {};
		\node[vertex] (x) at (2,0) {};
		\node[vertex] (a) at (3,0) {};
		
		\node[below] at (p) {$\strut p$};
		\node[below] at (v) {$\strut v$};
		\node[below] at (x) {$\strut x$};
		\node[below] at (a) {$\strut g$};
		
		\draw[nice dash] (p) -- (v) -- (x) -- (a);
		\end{tikzpicture}
		\hspace{20mm}
		\begin{tikzpicture}[xscale=0.5, yscale=0.5]
		\begin{scope}[local bounding box=T1]
		\node[vertex] (g) at (4,4) {};
		\node[vertex] (p) at (1,3) {};
		\node[vertex] (v) at (2,2) {};
		\node[vertex] (x) at (3,1) {};
		\draw (g) -- (p) -- (v) -- (x);
		\node[right] at (g) {$g$};
		\node[below left] at (p) {$p$};
		\node[below left] at (v) {$v$};
		\node[below left] at (x) {$x$};
		\end{scope}
		\begin{scope}[shift={(6,0.5)}, local bounding box=T2]
		\node[vertex] (g) at (4,3) {};
		\node[vertex] (p) at (1,1) {};
		\node[vertex] (v) at (2,2) {};
		\node[vertex] (x) at (3,1) {};
		\draw (g) -- (v) -- (p);
		\draw (v) -- (x);
		\node[right] at (g) {$g$};
		\node[below left] at (p) {$p$};
		\node[above left] at (v) {$v$};
		\node[below right] at (x) {$x$};
		\end{scope}
		\draw[-latex] (5,2.5) -- (6,2.5);
		\end{tikzpicture}
		\caption{The situation in \cref{p:impl-rot} if $v$ is a direct separator with a direct separator child $x$ in~$T$. Left: The path in the underlying graph. Right: The rotation at $v$.}\label{fig:rot-spec-1}
	\end{figure}
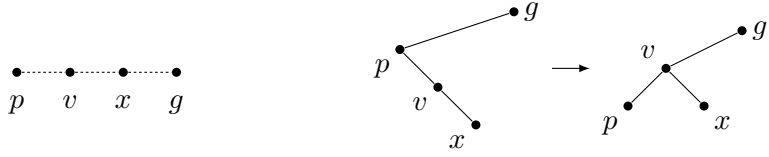
	
	Now consider the child pointers of $v$. Note that $v$ loses its parent and keeps all other ancestors, gains a child ($p$) and possibly loses a child ($c$). If $g = \bot$, then $v$ is the root of $T'$, and thus $\tdsepchild'(v) = \tisepchild'(v) = \bot$. Suppose $g \neq \bot$. Since $T$ is 2-cut, $v, p, g$ lie on a common path (using \cref{p:Steiner-closed-three-path}). The vertex $g$ cannot lie between $v$ and $p$ on this path, otherwise $v$ and $p$ would be in different subtrees of $g$ in~$T$.
	\begin{itemize}
		\item If $v$ is a direct separator in $T$, then $\bd(T_v) = \{p,g\}$ and $v$ lies on the path between $p$ and $g$. We claim that $\tdsepchild'(v) = \tisepchild(v)$. Indeed, if $x = \tdsepchild'(v) \neq \bot$, then $\bd(T'_x) = \{v,g\}$. Since $v$ is on the path between $p$ and $g$, we have $x \neq p$ (see \cref{fig:rot-spec-1}), so $x$ must already have been a child of $v$ in $T$, and $\bd(T_x) = \bd(T'_x) = \{v,g\}$, so $x = \tisepchild(v)$. Conversely, if $y = \tisepchild(v) \neq \bot$, then $y$ is still a child of $v$ in $T'$, so $\bd(T'_y) = \bd(T_y) \subseteq \bd(T_v) \cup \{v\} \setminus \{p\} = \{v,g\}$, implying that $\tdsepchild'(v) = y$.
		
		To determine $\tisepchild'(v)$, consider the following two cases.
		\begin{itemize}
			\item $p$ was a separator node in $T$. Then $\bd(T_p) = \{g, a\}$, where $a$ is some ancestor of $g$. \Cref{p:Steiner-closed-three-path} implies that $p$ is on the path between $g$ and $a$. Since $v$ lies between $g$ and $p$, the underlying tree has a path containing $a, p, v, g$, in that order (see \cref{fig:rot-spec-2}). Hence, $p$ lies on the path between $a$ and $v$, so $\bd(T'_p) = \{v,a\}$ and thus $\tisepchild'(v) = p$.
			\item $p$ was a 1-cut node in $T$. We claim that then $\tisepchild'(v) = \bot$. Suppose otherwise that $x = \tisepchild'(v) \neq \bot$. Then $\bd(T'_x) = \{v,a\}$, where $a$ is some ancestor of $g$. This implies $a \in \bd(T'_v) = \bd(T_p)$. But then $\bd(T_p) = \{g,a\}$, contradicting the assumption.
		\end{itemize}
		
		\item If $v$ is not a direct separator in $T$, then $p$ is on the path between $v$ and $g$. Thus, $\tdsepchild'(v) = p$. If $x = \tisepchild(v) \neq \bot$, then $x$ is still a separator child of $v$ (by definition). Since $v$ only gains $p$ as a child, no other nodes can become the indirect separator child of $v$. Thus, $\tisepchild'(v) = \tisepchild(v)$.
	\end{itemize}
	
	Finally, consider a separator child $x$ of $c$. Since $c$ swapped parent ($v$) and grandparent ($p$), if $x$ was a direct separator in $T$, it is an indirect separator in $T'$, and vice versa. Hence, we have $\tisepchild'(c) = \tdsepchild(c)$ and $\tdsepchild'(c) = \tisepchild(c)$.
	
	All nodes other than $v$, $p$, $g$, $c$ do not gain or lose children and do not change parent, hence their pointers are the same in $T$ and $T'$. Implementing a rotation procedure based on the observations above is straight-forward. (See \implfile{stt/src/twocut/basic.rs} in the source code.)
\end{proof}

\begin{figure}[t]
	\centering
	\begin{tikzpicture}[xscale=0.7]
		\node at (0,-1) {};
		\node at (0,1) {};
		
		\node[vertex] (a) at (0,0) {};
		\node[vertex] (p) at (1,0) {};
		\node[vertex] (v) at (2,0) {};
		\node[vertex] (g) at (3,0) {};
	
		\node[below] at (a) {$\strut a$};
		\node[below] at (p) {$\strut p$};
		\node[below] at (v) {$\strut v$};
		\node[below] at (g) {$\strut g$};
		
		\draw[nice dash] (a) -- (p) -- (v) -- (g);
	\end{tikzpicture}
	\hspace{20mm}
	\begin{tikzpicture}[xscale=0.5, yscale=0.5]
		\begin{scope}[local bounding box=T1]
			\node[vertex] (a) at (1,4) {};
			\node[vertex] (g) at (4,3) {};
			\node[vertex] (p) at (2,2) {};
			\node[vertex] (v) at (3,1) {};
			\draw[nice dash] (a) -- (g);
			\draw (g) -- (p) -- (v);
			\node[below left] at (a) {$a$};
			\node[right] at (g) {$g$};
			\node[below left] at (p) {$p$};
			\node[below left] at (v) {$v$};
		\end{scope}
		\begin{scope}[shift={(6,0)}, local bounding box=T2]
			\node[vertex] (a) at (1,4) {};
			\node[vertex] (g) at (4,3) {};
			\node[vertex] (p) at (2,1) {};
			\node[vertex] (v) at (3,2) {};
			\draw[nice dash] (a) -- (g);
			\draw (g) -- (v) -- (p);
			\node[below left] at (a) {$a$};
			\node[right] at (g) {$g$};
			\node[below right] at (p) {$p$};
			\node[below right] at (v) {$v$};
		\end{scope}
		\draw[-latex] (5,2.5) -- (6,2.5);
	\end{tikzpicture}
	\caption{The situation in \cref{p:impl-rot} and \cref{sec:monoid-weights} if $v$ is a direct separator and $p$ is a separator in $T$. Left: The path in the underlying graph. Right: The rotation at $v$.}\label{fig:rot-spec-2}
\end{figure}
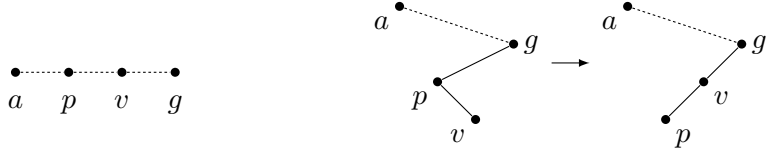

\section{Monoid edge weights}\label{sec:monoid-weights}

In this section, we show how to maintain more general \emph{monoid} weights. The difference between groups and monoids is that not every monoid element has an inverse. Thus, subtraction is not possible, so the approach from \cref{sec:edge-weights} cannot be used.

Let $T$ be a 2-cut forest on a tree $G$ with edge weights from a commutative monoid $(W,+)$. Again, let $d(u,v)$ denote the weight of the path between two nodes $u$ and $v$ (i.e., the \emph{distance} between $u$ and $v$), and let $d(u,v) = \infty$ if $u = \bot$ or $v = \bot$. We now need two fields for every node.
\begin{itemize}
	\item $\tpdist(v)$ is the distance between $v$ and its parent, or $\infty$ if $v$ is the root.
	\item $\tadist(v)$ is $\infty$ if $v$ is 1-cut. If $v$ is a separator, then $\tadist(v)$ is the distance between $v$ and the node $x \in \bd(T_v)$ that is \emph{not} the parent of $v$.
\end{itemize}

Together, $\tpdist(v)$ and $\tadist(v)$ store the distance of $v$ to each node $x \in \bd(T_v)$. We now show how to maintain both fields under rotations.

Consider a rotation of $v$ with its parent $p$. Let $T$, $T'$ be the search forest before and after the rotation. Let $c$ be the direct separator child of $v$ in $T$, or $\bot$ if $v$ has no direct separator child. We denote by $\tpdist(\cdot), \tadist(\cdot)$ and $\tpdist'(\cdot), \tadist'(\cdot)$ the values before and after the rotation.

Suppose $c$ exists. Then the rotation exchanges parent and grandparent of $c$, both of which are in $\bd(T_c) = \bd(T'_c)$. Hence, we can simply swap $\tpdist(c)$ and $\tadist(c)$.

As in the group setting, we always have $\tpdist'(p) = d(p,v) = \tpdist(v)$. If $p$ is the root of $T$, then $v$ is the root of $T'$. We then have $\tpdist'(v) = \tadist'(v) = \tadist'(p) = \infty$.

We now turn to the more interesting case where $p$ is not the root. Then, $p$ has a parent $g$. Let $a$ be the other node in $\bd(T_p)$ if $p$ is a separator, or $\bot$ otherwise. We have $\tpdist'(v) = d(v,g)$ and $\tadist'(v) = d(v,a)$.

\Cref{p:Steiner-closed-three-path} implies that $v,p,g$ lie on a common path.
\begin{itemize}
	\item If $v$ is between $p$ and $g$ on the path, then $v$ is a direct separator in $T$ (by definition). We have $\tpdist'(v) = d(v,g) = \tadist(v)$.
	\begin{itemize}
		\item If $p$ is a separator in $T$, then $p$ lies on a path between $g$ and $a$ in $G$, which means $G$ has a path along $a, p, v, g$, in that order (see \cref{fig:rot-spec-2}). Thus, $\tadist'(p) = d(p,a) = \tadist(p)$ and $\tadist'(v) = d(v,a) = d(v,p) + d(p,a) = \tpdist(v) + \tadist(p)$.
		\item If $p$ is not a separator in $T$, then $a = \bot$, so $\tadist'(v) = d(v,a) = \infty$. We claim that $p$ is not a separator in $T'$ either, so $\tadist'(p) = \infty$. Indeed, if $p$ is a separator in $T'$, then $\bd(T'_p) \subseteq \{v\} \cup \bd(T'_v) = \{v,g\}$ (since $\bd(T'_v) = \bd(T_p) = \{g\}$). But this contradicts the assumption that $v$ is on the path between $p$ and $g$.
	\end{itemize}
	\item If $p$ is between $v$ and $g$ on the path, then $\tpdist'(v) = d(v,g) = d(v,p) + d(p,g) = \tpdist(v) + \tpdist(p)$ as in the group case. Also, $\bd(T'_p) = \{v,g\}$, so $\tadist'(p) = d(p,g) = \tpdist(p)$.
	
	It remains to compute $\tadist'(v) = d(v,a)$. If $p$ is not a separator in $T$, then $a = \bot$ and thus $\tadist'(v) = \infty$. If $p$ is a separator in $T$, then $v$ must also be a separator in $T$ (otherwise, the rotation is not valid by \cref{p:rot-allowed}). We have $\bd(T_v) \subseteq \{p\} \cup \bd(T_p) = \{p,g,a\}$ by \cref{p:boundary-rules}. Since $p$ is between $v$ and $g$ by assumption, we have $g \notin \bd(T_v)$, and thus $\bd(T_v) = \{p,a\}$, implying that $\tadist'(v) = \tadist(v)$.
	\item $g$ cannot be between $v$ and $p$, since then $v$ and $p$ needed to be in different subtrees of $g$.
\end{itemize}

\section{Stability}\label{sec:stability}

In this section, we argue that \algMRT{} and the three SplayTT variants are stable. Given an STT $T$ and a node $v \in V(T)$, a node $x \in V(T)$, we define the property $P(T,v,x)$ as satisfied if
\begin{itemize}
	\item all nodes on the root path of $x$ are 1-cut; and
	\item the root paths of $v$ and $x$ only intersect at the root of $T$.
\end{itemize}

\begin{lemma}\label{p:stable-rotation}
	Let $T$ be an STT, let $v, x \in V(T)$, and let $T'$ be the result of rotating at a node $u$ such that
	\begin{enumerate}[(i)]
		\item $u$ is on the root path of $v$; and
		\item if $u$ is a child of the root and $u \neq v$, then the child of $u$ on the root path of $v$ is 1-cut.
	\end{enumerate}
	
	Then, $P(T,v,x)$ implies $P(T',v,x)$.
\end{lemma}
\begin{proof}
	Let $p$ be the parent of $u$ in $T$. First, suppose that $p$ is not the root. Then $P(T,v,x)$ implies that $p$ is not on the root path of $x$. Hence, rotating at $u$ may remove $p$ from the root path of $v$, but does not change the root path of $x$. Thus, $P(T',v,x)$ holds.
	
	Now suppose that $p$ is the root of $T$. Then rotating at $u$ adds $u$ to the root path of~$x$. We have $\bd(T'_u) = 0$, $\bd(T'_p) = 1$, and $\bd(T'_y) = \bd(T_y) = 1$ for every other node $y$ on the root path of $x$. This proves the first part of $P(T',v,x)$. In order to prove the second part, let $v'$ be the child of $u$ on the root path of $v$. Assumption (ii) implies that $v'$ is 1-cut, so rotating at $u$ does not change the parent of $v'$, and thus rotating at $u$ removes $p$ from the root path of $v$. Hence, the second part of $P(T',v,x)$ holds.
\end{proof}

We now argue that all four algorithms only perform rotations satisfying \cref{p:stable-rotation}. Clearly, all rotations are performed on the root path of $v$. To see that (ii) holds, we need to consider the algorithms in more detail. We are only interested in rotations at a node $u \neq v$ that is the child of the root. Such a rotation only happens in the following circumstances.
\begin{itemize}
	\item The second-to-last rotation in \algMRT, when we rotate at the parent of $v$. This only happens if $v$ is 1-cut, and thus satisfies (ii).
	\item A ZIG-ZIG step at a grandchild $u'$ of the root $r$, in any of the SplayTT variants. A ZIG-ZIG step only happens if $u', u, r$ are on a path in $G$ in that order, implying that $u'$ is 1-cut before the rotation.
	\item A final ZIG step in the first pass of (Local) Two-pass SplayTT, when bringing the final branching node to the root. In that case, $u$ was a branching node before rotating it to the root, so by definition, the child of $u$ on the root path of $v$ is 1-cut.
\end{itemize}

Let $T$ be an STT with root $r$, and let $T'$ be the result of calling $\tNodeToRoot{v}$ (using any of the four algorithms). Then $P(T,v,r)$ trivially holds, and \cref{p:stable-rotation} implies that $P(T',v,r)$ also holds.

To prove stability, it remains to show that the depth of $r$ in $T'$ is bounded. Since $P(T'',v,r)$ holds for every intermediate tree, the depth of $r$ can only increase when a rotation involving the (current) root is performed. It is easy to see that each variant performs at most three rotations or $\tsplayStep$ calls that involve the root, so the final depth of $r$ is at most six. We thus conclude that \algMRT{} and all three SplayTT variants are stable.

\section{Maintaining rooted trees with STTs}\label{sec:stts-rooted}

In contrast to link-cut trees, our implementation of STTs (as described in \cref{sec:2cut-impl}) cannot represent rooted trees without modification. Consider a call to $\opFindRoot{v}$ when $v$ is the root of its STT. Since $v$ has no separator children, both its child pointers are $\bot$, thus we cannot navigate to the root of the underlying tree (or \emph{any} node other than $v$, for that matter).
In this section, we show how we can implement $\opFindRoot{v}$ and other rooted-tree operations using extra data. The relevant part of the source code is found in \implfile{stt/src/twocut/rooted.rs}.

Let $G$ be a tree with a designated root $r$. Let $T$ be an STT on $G$. We store (a pointer to) the root $r$ in each node on the root path of $r$ in $T$. Formally, we maintain the following property for each node $v \in V(T)$:
\begin{align*}
	\tdroot(v) = \begin{cases}
		r, &\text{if } r \in V(T_v)\\
		\bot, &\text{otherwise}
	\end{cases}
\end{align*}

Implementing $\opFindRoot{v}$ is now trivial: Call $\tNodeToRoot{v}$, and then return $\tdroot(v)$.

We now describe how to update $\tdroot(v)$ under rotations. Consider a rotation of $v$ with its parent $p$. Let $T$, $T'$ denote the search tree before and after the rotation and let $\tdroot$, $\tdroot'$ denote the respective values before and after the rotation.

Observe that only $\tdroot(v)$ and $\tdroot(p)$ may change. Since $V(T'_v) = V(T_p)$, we have $\tdroot'(v) = \tdroot(p)$. For $\tdroot'(p)$, consider the following cases.
\begin{itemize}
	\item If $\tdroot(p) = \bot$, then $\tdroot'(p) = \bot$, since $p$ gains no new descendants with the rotation.
	\item If $\tdroot(p) \neq \bot$ and $\tdroot(v) = \bot$, then $r$ is in $V(T_c)$ for some child $c \neq v$ of $p$. Observe that $c$ is still a child of $p$ in $T'$, hence $\tdroot'(p) = \tdroot(p)$.
	\item Finally, suppose $\tdroot(p) \neq \bot$ and $\tdroot(v) \neq \bot$. Let $c$ be the direct separator child of $v$. If $c$ does exist and $\tdroot(c) \neq \bot$, then $r$ is in $V(T_c) = V(T'_c)$ and hence in $V(T'_p)$, so $\tdroot'(p) = \tdroot(p)$. Otherwise, $r = v$ or $r \in V(T_{c'})$ for some child $c' \neq c$ of $v$ in $T$, hence $\tdroot'(p) = \bot$.
\end{itemize}

We now sketch the implementations of the remaining operations. We refer to the code in \implfile{stt/src/twocut/rooted.rs} for more details.

$\opLink{u,v}$ works as described in \cref{sec:link-cut}, except that afterwards we set $\tdroot(u) \gets \bot$. Note that, by assumption, $u$ was the root of its underlying tree before the operation. After bringing $u$ to the root, each descendant $x$ of $u$ has $\tdroot(x) = \bot$. Then, $u$ becomes the child of $v$, so it is no longer the underlying tree root.

$\opCut{v}$ now only takes one parameter. If the parent $u$ of $v$ (in the underlying tree) is known, we can simply call $\tNodeToRoot{v}$ and $\tNodeBelowRoot{u}$, then detach $u$ from $v$ and set $\tdroot(v) \gets v$. However, we do not assume that $u$ is known. To find $u$, after calling $\tNodeToRoot{v}$, we first call $\tNodeBelowRoot{r}$ (note that $r = \tdroot(v)$). Then, we can find $u$ by first moving to the direct separator child of $r$, and then following indirect separator child pointers as long as possible. This moves along the underlying path from $r$ to $v$ (possibly skipping nodes), stopping at $u$. If $r$ has no direct separator child, then $u = r$.

For $\opLCA{u,v}$, we first call $\tNodeToRoot{v}$ and $\tNodeBelowRoot{u}$. Some simple checks determine whether $u$ is an ancestor of $v$ or vice versa. Otherwise, the LCA has to be in the direct separator child $x$ of $u$. Now, we check if $\tdroot(d) \neq \bot$ or $\tdroot(i) \neq \bot$, where $d$ and $i$ are the direct and indirect separator children of $x$. If either is true, we repeat with $x \gets i$, resp., $x \gets d$. Otherwise, it can be seen that $x$ must be the LCA of $u$ and $v$. Calling $\tNodeToRoot{x}$ at the end pays for following the root path to $x$ (via amortization).

Finally, $\opEvert{v}$ is implemented as follows. First, call $\tNodeToRoot{v}$. We then need to set $\tdroot(v) \gets\nobreak v$ and $\tdroot(x) \gets \bot$ for each node $x \neq v$. Observe that every node $x$ with $\tdroot(x) \neq \bot$ must be on the root path of $r$, hence we can make the necessary changes by following parent pointers from $r$. To pay for this, we call $\tNodeToRoot{r}$ afterwards.

\section{More detailed experimental results}\label{sec:detailed-results}

See \cref{tab:full_queries,tab:full_mst,tab:full_lca,tab:full_lca_evert} on \cpagerefrange{tab:full_queries}{tab:full_lca_evert}

\begin{table*}[b]
	\centering\small
		\begin{tabular}{lcccccccc}
		\toprule
		& \multicolumn{8}{c}{
Running time (\textmu{}s/query)}\\		\cmidrule{2-9}Algorithm & $n = 1000$ & $2000$ & $3000$ & $4000$ & $5000$ & $6000$ & $7000$ & $8000$\\
		\midrule
		Petgraph & $7.64$ & -- & -- & -- & -- & -- & -- & --\\
		Link-cut & $0.48$ & $0.53$ & $0.56$ & $0.58$ & $0.60$ & $0.62$ & $0.63$ & $0.65$\\
		Greedy Splay & $0.39$ & $0.44$ & $0.47$ & $0.48$ & $0.50$ & $0.51$ & $0.53$ & $0.54$\\
		Stable Greedy Splay & $0.36$ & $0.42$ & $0.44$ & $0.46$ & $0.48$ & $0.49$ & $0.50$ & $0.52$\\
		2P Splay & $0.41$ & $0.47$ & $0.49$ & $0.51$ & $0.53$ & $0.54$ & $0.55$ & $0.56$\\
		Stable 2P Splay & $0.37$ & $0.43$ & $0.46$ & $0.47$ & $0.49$ & $0.50$ & $0.51$ & $0.53$\\
		L2P Splay & $0.38$ & $0.43$ & $0.46$ & $0.48$ & $0.50$ & $0.51$ & $0.52$ & $0.54$\\
		Stable L2P Splay & $0.36$ & $0.42$ & $0.45$ & $0.46$ & $0.48$ & $0.49$ & $0.51$ & $0.52$\\
		MTR & $0.32$ & $0.36$ & $0.38$ & $0.39$ & $0.41$ & $0.42$ & $0.43$ & $0.44$\\
		Stable MTR & $0.31$ & $0.35$ & $0.37$ & $0.39$ & $0.40$ & $0.41$ & $0.43$ & $0.44$\\
		1-cut & $0.16$ & $0.27$ & $0.36$ & $0.46$ & $0.55$ & $0.64$ & $0.72$ & $0.77$\\
		\bottomrule
	\end{tabular}

	\caption{Results for the uniformly random connectivity queries experiment.}\label{tab:full_queries}
\end{table*}

\begin{table*}[b]
	\centering\small
		\begin{tabular}{lccccccc}
		\toprule
		& \multicolumn{7}{c}{
Running time (\textmu{}s/edge)}\\		\cmidrule{2-8}Algorithm & $n = 10000$ & $20000$ & $50000$ & $100000$ & $200000$ & $500000$ & $1000000$\\
		\midrule
		Kruskal (petgraph) & $0.09$ & $0.10$ & $0.14$ & $0.20$ & $0.32$ & $0.48$ & $0.63$\\
		Link-cut & $1.29$ & $1.41$ & $1.56$ & $1.90$ & $2.38$ & $3.00$ & $3.43$\\
		Greedy Splay & $1.03$ & $1.14$ & $1.27$ & $1.54$ & $1.97$ & $2.54$ & $2.94$\\
		Stable Greedy Splay & $1.01$ & $1.11$ & $1.24$ & $1.50$ & $1.92$ & $2.51$ & $2.90$\\
		2P Splay & $1.05$ & $1.16$ & $1.28$ & $1.54$ & $2.00$ & $2.64$ & $3.05$\\
		Stable 2P Splay & $1.00$ & $1.11$ & $1.23$ & $1.48$ & $1.95$ & $2.58$ & $2.99$\\
		L2P Splay & $1.05$ & $1.16$ & $1.29$ & $1.55$ & $2.00$ & $2.61$ & $3.01$\\
		Stable L2P Splay & $1.02$ & $1.13$ & $1.26$ & $1.53$ & $2.00$ & $2.61$ & $3.01$\\
		MTR & $0.84$ & $0.93$ & $1.03$ & $1.27$ & $1.69$ & $2.23$ & $2.59$\\
		Stable MTR & $0.85$ & $0.93$ & $1.04$ & $1.27$ & $1.70$ & $2.24$ & $2.60$\\
		1-cut & $0.61$ & $0.97$ & $1.52$ & $1.56$ & $2.16$ & $3.74$ & $5.32$\\
		\bottomrule
	\end{tabular}

	\caption{Results for the uniformly random incremental MSF experiment.}\label{tab:full_mst}
\end{table*}

\begin{table*}[b]
	\centering\small
		\begin{tabular}{lccccccc}
		\toprule
		& \multicolumn{7}{c}{
Running time (\textmu{}s/query)}\\		\cmidrule{2-8}Algorithm & $n = 10000$ & $20000$ & $50000$ & $100000$ & $200000$ & $500000$ & $1000000$\\
		\midrule
		Link-cut & $0.28$ & $0.31$ & $0.34$ & $0.36$ & $0.46$ & $0.67$ & $0.83$\\
		Greedy Splay & $0.38$ & $0.42$ & $0.45$ & $0.48$ & $0.59$ & $0.79$ & $0.94$\\
		2P Splay & $0.40$ & $0.43$ & $0.48$ & $0.51$ & $0.67$ & $0.94$ & $0.95$\\
		L2P Splay & $0.38$ & $0.41$ & $0.46$ & $0.49$ & $0.65$ & $0.92$ & $0.93$\\
		MTR & $0.31$ & $0.34$ & $0.37$ & $0.40$ & $0.54$ & $0.79$ & $0.85$\\
		Simple & $1.63$ & $2.59$ & -- & -- & -- & -- & --\\
		\bottomrule
	\end{tabular}

	\caption{Results for the LCA experiment.}\label{tab:full_lca}
\end{table*}

\begin{table*}[b]
	\centering\small
		\begin{tabular}{lccccccc}
		\toprule
		& \multicolumn{7}{c}{
Running time (\textmu{}s/query)}\\		\cmidrule{2-8}Algorithm & $n = 10000$ & $20000$ & $50000$ & $100000$ & $200000$ & $500000$ & $1000000$\\
		\midrule
		Link-cut & $0.50$ & $0.55$ & $0.61$ & $0.65$ & $0.79$ & $1.11$ & $1.34$\\
		Greedy Splay & $0.47$ & $0.51$ & $0.55$ & $0.59$ & $0.71$ & $0.95$ & $1.12$\\
		2P Splay & $0.48$ & $0.52$ & $0.57$ & $0.60$ & $0.72$ & $0.96$ & $1.13$\\
		L2P Splay & $0.47$ & $0.51$ & $0.56$ & $0.59$ & $0.72$ & $0.96$ & $1.13$\\
		MTR & $0.36$ & $0.40$ & $0.43$ & $0.46$ & $0.59$ & $0.83$ & $1.00$\\
		Simple & $1.44$ & $2.17$ & -- & -- & -- & -- & --\\
		\bottomrule
	\end{tabular}

	\caption{Results for the LCA experiment with \opEvert{}.}\label{tab:full_lca_evert}
\end{table*}

\end{document}